\documentclass[journal]{IEEEtran}

\usepackage{cite}
\usepackage{xcolor}
\ifCLASSINFOpdf
\usepackage[pdftex]{graphicx}
\else
\fi
\usepackage{array}
\usepackage{mathrsfs}
\usepackage{enumitem}
\usepackage{algorithm} 
\usepackage{algpseudocode}
\usepackage[outdir=./]{epstopdf}
\usepackage{graphicx}
\usepackage{amssymb}
\usepackage{caption}
\usepackage{subcaption}
\usepackage{kantlipsum} %<- For dummy text
\usepackage{graphicx}
\usepackage{textcomp}
\usepackage{cite}
\usepackage{subcaption}
\usepackage[font={small}]{caption}
\usepackage{algorithm}
\usepackage{algpseudocode}
\makeatletter
\def\BState{\State\hskip-\ALG@thistlm}
\makeatother
\usepackage{array}
\newcolumntype{P}[1]{>{\centering\arraybackslash}p{#1}}
\usepackage{color}
\usepackage{graphicx}
\usepackage{amsmath}
\usepackage{amssymb}
\usepackage{amsthm}
\newcommand\numberthis{\addtocounter{equation}{1}\tag{\theequation}}
\newtheorem{prop}{\hspace{0.5cm}Proposition}
\newtheorem*{proof*}{\hspace{0.8cm}{\it{Proof}}}

\ifCLASSINFOpdf
\else
\fi
\usepackage{url}
\usepackage{subcaption}
\usepackage[font={small}]{caption}
\makeatletter
\def\BState{\State\hskip-\ALG@thistlm}
\makeatother
\newcolumntype{P}[1]{>{\centering\arraybackslash}p{#1}}

\ifCLASSINFOpdf
\else
\fi
\usepackage{url}

\begin{document}
%
% paper title
% Titles are generally capitalized except for words such as a, an, and, as,
% at, but, by, for, in, nor, of, on, or, the, to and up, which are usually
% not capitalized unless they are the first or last word of the title.
% Linebreaks \\ can be used within to get better formatting as desired.
% Do not put math or special symbols in the title.
\title{ Device vs Edge Computing for Mobile Services: Delay-aware Decision Making to Minimize Power Consumption}
%
%
% author names and IEEE memberships
% note positions of commas and nonbreaking spaces ( ~ ) LaTeX will not break
% a structure at a ~ so this keeps an author's name from being broken across
% two lines.
% use \thanks{} to gain access to the first footnote area
% a separate \thanks must be used for each paragraph as LaTeX2e's \thanks
% was not built to handle multiple paragraphs
%
\author{Meysam~Masoudi, %~\IEEEmembership{Student Member,~IEEE,}
Cicek~Cavdar %,~\IEEEmembership{Member,~IEEE}% <-this % stops a space
\\
School of Electrical Engineering and Computer Science, KTH Royal Institute of Technology,\\ Email: 
\{masoudi,cavdar\}@kth.se 
\thanks{ Part of this work has been published in IEEE WCNC2017 \cite{Maso1703:Green}. }% <-this % stops a space \cite{Maso1703:Green}
\thanks{This study is supported by EU Celtic Plus Project SooGREEN Service Oriented Optimization of GREEN mobile networks.}% <-this % stops a space
%\thanks{Manuscript received April 19, 2005; revised September 17, 2014.}
}
\maketitle
% As a general rule, do not put math, special symbols or citations
% in the abstract or keywords.
\begin{abstract}
A promising technique to provide mobile applications with high computation resources is to offload the processing task to the cloud. Utilizing the abundant processing capabilities of the clouds, mobile edge computing enables mobile devices with limited batteries to run resource hungry applications  and to save power. However, it is not always true that edge computing consumes less power compared to device computing. It may take more power for the mobile device to transmit a file to the cloud than running the task itself. This paper investigates the power minimization problem for the mobile devices by data offloading in multi-cell multi-user OFDMA mobile edge computing networks. We consider the maximum acceptable delay as QoS metric to be satisfied in our network. We formulate the problem as a mixed integer nonlinear problem which is converted into a convex form using D.C. approximation. To solve the   \textcolor{black}{converted} optimization problem, we have proposed centralized and distributed algorithms for joint power allocation and channel assignment together with decision-making. % \textcolor{black}{The proposed algorithms can solve the problem sub-optimally.} 
Simulation results illustrate that by utilizing the proposed algorithms, considerable power savings can be achieved, e.g., about $ 60 \% $ for  large bit stream size compared to local computing  baseline.
\end{abstract}
% Note that keywords are not normally used for peerreview papers.
\begin{IEEEkeywords}
Offloading, Resource Allocation, Mobile Cloud Computing, Mobile Edge Computing.
\end{IEEEkeywords}
% For peer review papers, you can put extra information on the cover
% page as needed:
% \ifCLASSOPTIONpeerreview
% \begin{center} \bfseries EDICS Category: 3-BBND \end{center}
% \fi
%
% For peerreview papers, this IEEEtran command inserts a page break and
% creates the second title. It will be ignored for other modes.
\IEEEpeerreviewmaketitle
\section{Introduction}
% The very first letter is a 2 line initial drop letter followed
% by the rest of the first word in caps.
% 
% form to use if the first word consists of a single letter:
% \IEEEPARstart{A}{demo} file is ....
% 
% form to use if you need the single drop letter followed by
% normal text (unknown if ever used by IEEE):
% \IEEEPARstart{A}{}demo file is ....
% 
% Some journals put the first two words in caps:
% \IEEEPARstart{T}{his demo} file is ....
% 
% Here we have the typical use of a "T" for an initial drop letter
% and "HIS" in caps to complete the first word.
\IEEEPARstart{S}{wift} growth in the development of \textcolor{black}{resource-hungry} mobile applications has motivated users to use smart phones as a platform for running the applications. However, mobile devices cannot always be considered as a platform for \textcolor{black}{resource-hungry} applications due to their limited power and processing capacity. Moreover, one of the key concerns of users is the battery lifetime of mobile devices \cite{kumar2010cloud} while running the applications, knowing the fact that increasing the clock frequency of a CPU increases the power consumption \cite{kwak2015dream}. Therefore, there is a tension between the \textcolor{black}{resource-hungry} applications and mobile devices \textcolor{black}{with limited battery and processing capacities}. To tackle the aforementioned problem, one solution is to bridge the gap between available and required resources by offloading the burden from mobile devices to the cloud \cite{kaewpuang2013framework}. 

%Processing at the cloud, which is equipped with the abundant processing resources, has become an attractive solution in order to ease this pain for the storage and data processing. Edge computing for mobile applications will enable new services for mobile users. 
\textcolor{black}{
	Mobile cloud computing (MCC) provides infrastructure, platform, and software as services to  mobile users \cite{xu2013survey}. Alongside  MCC,  mobile edge computing (MEC) has been proposed to bring the computation resources even closer to the mobile users. In MCC, the cloud is equipped with farms of computers, and it is considered as a fully centralized approach to provide such a service. On the other hand, MEC is supposed to be deployed in a distributed manner \cite{2017arXiv170205309M}.  Since these services are provided for the mobile users, the interaction between edge cloud and mobile users is inevitable. Consequently, once users decide to offload data to the cloud, it is necessary to  utilize the available resources efficiently. Otherwise, users cannot benefit from the potential advantages of offloading.}

It is true that edge computing can potentially save power for the mobile users \cite{kumar2010cloud}; however, this is not always true when the device consumes more power to transmit  data to the cloud than to process that data itself \cite{zhang2015offloading}. Because of the interference and radio channel conditions, the transmission of  data may consume more power for the mobile device. However it is not trivial to decide after making a simple comparison of two power figures for each device served by one base station since the decision may create interference and change the channel conditions for neighboring devices in the surrounding cells. There is  another important factor which has an impact on the decision: delay. A decision making procedure must consider the delay sensitivity of the applications to determine whether to choose local processing or offloading. Mobile devices consume more power as the delay requirement gets more stringent to process a certain task \cite{zhang2013energy}. \textcolor{black}{
 Example applications that can highly benefit from collaboration between mobile devices and cloud platform \cite{tran2019joint} includes online gaming \cite{basiri2018delay}, face recognition and detection \cite{soyata2012cloud}, healthcare \cite{hosseini2017deep}, tele-surgery \cite{nunna2015enabling}. Among these applications, some are more sensitive to delay than others, for instance, the interaction games, e.g., action and racing games, are more sensitive to the delay in comparison with puzzle and strategy games \cite{basiri2018delay}.} 
Delay requirements of different mobile broadband services can be seen in Table \ref{DelayReq}. 

In this paper, we investigate the power saving potential of data offloading in mobile devices under a multi-cell multi-user scenario and propose efficient algorithms to make decisions simultaneously for mobile devices to minimize the total power consumption by meeting the delay requirements from the services. Channel assignment and power allocation problems are considered jointly with the offloading-decision. 
\begin{table}[ht]
\centering
% \scriptsize
% \addtolength{\tabcolsep}{-3pt}
\bgroup
\caption{Acceptable delay for different services} \label{DelayReq}
\def\arraystretch{1.2}% 1 is the default, change whatever you need
\begin{tabular}{ m{3.5cm} l }
\hline
\textbf{Service Type} & \textbf{Acceptable Delay \cite{claypool2006latency,dusi2012closer,skorin2010analysis}}\\ \hline
\textbf{Online Games} & $ < 1000 $ $ ms $ \\ 
\hspace{10pt} Omnipresent & $ 1000 $ $ ms $ \\ 
\hspace{10pt} Third person avatar & $ 500 $ $ ms $\\ 
\hspace{10pt} First person avatar & $ 100 $ $ ms $\\ 
\textbf{Audio services} & $ < 450$ $ms $ \\ 
\hspace{10pt} Voice over IP & $ 200 $ $ ms $ \\ 
\textbf{Video Services} & $ < 150$ $ ms $ \\ 
\hspace{10pt} Video over IP & $ 70 $ $ ms $ \\
\textbf{Data} & $< 400$ $ ms $ \\ 
% \hspace{10pt} 1 & $ 100 $ $ ms $ \\ 
\textbf{Medical Data Transfer} & $ 100-400$ $ ms $ \\ 
\textbf{Tele-surgery} & $ 300$ $ ms $ \\ 
\textbf{Electrocardiogram } & $ \approx 1000$ $ ms $ \\
\textbf{Non real-time services} & Few seconds \\ 
% Macrocell radius &$ R_{M} $& $ 1000 $ $ m $ \\ 
% Femtocell radius &$ R_{F} $& $ 20 $ $ m $ \\ 
% & \multicolumn{4}{c||}{\bf{Sub-channel values }} & \multicolumn{4}{c|}{\bf{Power values}} \\
% & {\bf{Sub-channel Assignment}} & {\bf{Data Offloading}} & {\bf{Power Allocation}} \\
% \hline
% {\bf{Approach I}} & $O(NFN_{c})$ & \multicolumn{2}{c|}{$O \bigg(N_c^4 F^3 N^3 (3F + N ) \bigg)$} \\ 
% \hline
% \bf{Approach II} & $O(NFN_{c})$ & $O(N_c F)$ & $O \bigg(N_c^4 F^3 N^3 (2F+N)\bigg)$ \\
% \hline
\end{tabular}
% \addtolength{\tabcolsep}{3pt}
\egroup
\end{table}
\subsection{Related Works }
There are a couple of issues to be addressed in the context of MCC and MEC, namely, the architecture, the power consumption of the network, and delay.

The surveys in \cite{2017arXiv170205309M} and \cite{DBLP:journals/corr/MaoYZHL17},  \textcolor{black}{studied the state of the art on the integration of}  MEC to the mobile networks, the computation offloading schemes, resource management problems, and their current challenges. The feasibility of mobile computation offloading is investigated in \cite{barbera2013offload} using experimental measurements. The authors in \cite{Tout2017}  proposed \textcolor{black}{an architecture for mobile computation offloading} where  device resources, e.g., energy and CPU usage \textcolor{black}{can be saved}. An  application offloading framework to the cloud is proposed and implemented in \cite{kosta2012thinkair} \textcolor{black}{where}  execution time and energy consumption of devices \textcolor{black}{ are reduced}. The authors in \cite{mahmoodi2015cloud} modeled the energy consumption of the mobile devices. For a single device, they formulated an energy minimization problem considering computation offloading to the cloud. The study in \cite{ellouze2015mobile} presents an offloading decision model to  extend the battery of single mobile device.  The authors in \cite{7425262}, modeled the offloading decision as a competitive game \textcolor{black}{considering multiple mobile devices where users try to minimize their energy consumption}. \textcolor{black}{These studies} did not consider the power allocation which has significant impact on the performance of their algorithms.  The authors in \cite{liu2014effective} solved the offloading decision problem \textcolor{black}{and reduced} the time complexity of the application offloading \textcolor{black}{aiming} to remove the processing burden from mobile devices; \textcolor{black}{However, they did not  consider the impact of} resource allocation.

%In \cite{satyanarayanan2009case}, the authors dealt with the latency issue by means of cloudlet infrastructure, which is a data center to bring the cloud closer to the users. 

%The authors in \cite{im2016amuse}, presented a practical offloading framework in a cost aware Wi-Fi system considering the throughput-delay trade offs. 
% 
% To find the optimal offloading policy, a Markov decision process is adopted in \cite{zhang2015offloading}. 

%energy

Resource management schemes are  key techniques to  minimize the power consumption while guaranteeing  the quality of service (QoS) which is critical for the MEC networks \cite{piunti2015energy}. Accordingly in \cite{chen2016joint}, a heuristic based resource allocation approach is adopted to minimize the energy consumption of all users while making decision on offloading. In \cite{7307234}, a game theoretic approach is adopted to design an offloading mechanism for mobile devices. \textcolor{black}{Although}, a multi-user \textcolor{black}{scenario is } considered \textcolor{black}{, QoS and counter-impact between different users for service degradation is not taken into account.} In \cite{chen2015decentralized}, a decentralized offloading game is proposed to make decisions among mobile devices in a simple single channel scenario. The partitioning problem for a mobile data streaming application is defined  and solved by a genetic algorithm  in \cite{yang2013framework}. It is reported that computation partitioning between mobile and cloud  can enhance the application processing speed. In \cite{cheng2016computation}, the authors studied the problem of computation offloading in C-RAN based MCC network. In their study, the authors jointly optimized the beamforming design and power allocation with a decision making strategy to minimize the network energy consumption.

%delay
Along with power consumption, delay is a critial factor in MEC networks. The authors in \cite{DBLP:journals/corr/Mao0L17}   proposed an algorithm to optimize the power consumption and to minimize the delay. They considered a simple single-user MEC system. In this study, the interference analysis and its effect on the offloading decision is missing. The authors in \cite{7541539} modeled the problem of task offloading as a  Markov decision process and solved a delay minimization problem to find \textcolor{black}{an} optimal task scheduling policy. For the energy consumption and latency minimization problem, a partial computation offloading algorithm is proposed in \cite{wang2016mobile} to optimize the computational speed of mobile devices and their transmit power. In \cite{7850968}, to minimize the offloading energy consumption, the authors \textcolor{black}{studied} joint optimization of computing and radio resources considering the latency constraints in a MEC network. The authors in \cite{liu2016wireless} considered the problem of resource scheduling for multi-service multi-user MCC networks to minimize  the average queuing delay of the system. A dynamic task offloading algorithm using Lyapunov optimization is proposed in \cite{huang2012dynamic}, aiming at minimizing the energy consumption of one user with constraint on the maximum acceptable application delay. However, the interference and its impact on offloading decision is  \textcolor{black}{not taken into account}.  The authors in \cite{zhang2013energy} proposed a model for the mobile device energy consumption. They have derived an offloading policy considering both delay and energy consumption under a single stochastic wireless channel with only "good" or "bad" channel state. However, \textcolor{black}{this study} is limited with single-user single-channel \textcolor{black}{scenario, considering neither interference nor QoS.} The authors in \cite{8234573}  have focused on the \textcolor{black}{ trade-off between} latency and energy consumption in MEC network overlaid by small cells. They optimized the communication resources for computation offloading while considering delay sensitive tasks. In \cite{Maso1703:Green}, we proposed  \textcolor{black}{an algorithm for joint optimization of } power allocation, offloading decision making, and channel assignment (J-PAD)  to perform  resource allocation considering both interference and delay constraints in  multi-cell multi-user networks. \textcolor{black}{In preliminary version of our study, devices offload the processing task to the cloud aiming to minimize overall power consumption of devices}. \textcolor{black}{However, this approach was centralized and it is important to design decentralized algorithms for practical usage to enable mobile devices to make their decisions}.

\subsection{Contributions}
There are still plenty of challenges to be tackled in the multi-cell multi-user and multi-channel MEC networks. The \textcolor{black}{joint} problem of resource allocation and decision making for data offloading in \textcolor{black}{such networks considering both}  QoS and interference is still missing in the literature.
%First, how a user can make decision on offloading to save its energy consumption while satisfying the QoS requirements. Second, if a user decides to offload, how much resource should be allocated and what is the optimal procedure of allocating the resources to the corresponding users? To tackle the aforementioned challenges, 
In this paper, we aim to minimize the power consumption of users while considering the users' QoS in terms of maximum tolerable delay. We formulate the resource allocation and offloading optimization problem.  To have a tractable solution, we convert the problem into a convex form and propose two algorithms to solve the \textcolor{black}{convex} problem in a polynomial time. %called J-PAD and C-PAD
%First algorithm perform decision making and power allocation simultaneously so that the resources can be utilized in a more efficient way. Alternatively, in the second algorithm, a decision making criterion is proposed to enable the users to make decision on their own device with a limited needed information about the network. These algorithms are distributed scheme and therefore become more practical. 
%\textcolor{black}{text}

The main contributions of this paper can be summarized as follows: 
\begin{itemize}
\item[$\bullet$ ] In the context of multi-cell multi-user OFDMA MEC networks, we formulate the resource allocation and offloading problem that is aware of network status and users' demand aiming to minimize the total power consumption of all users subject to constraints on QoS.%, and on backhaul capacity.
\item[$\bullet$ ] We formulate the problem as a mixed integer nonlinear optimization problem (\textbf{MINLP}). To solve the problem, we convert it to the convex form using variable changing, DC approximation, adding a penalty factor, and relaxing the binary constraints. Therefore, the \textcolor{black}{converted} problem can be solved in a polynomial time. \textcolor{black}{Proven that it converges in a polynomial time, the proposed efficient solutions in terms of complexity, compared to the lower bound have  30 \% percent worst-case optimality gap}. % with a much lower complexity. 
 \textcolor{black}{To evaluate the optimality, we compared the solution with the lower bound that is obtained by solving the problem for  interference-free scenario. }
\item[$\bullet$ ] We also propose two algorithms to solve the problem of joint resource allocation and decision-making. 
%The first algorithm performs offloading decision and allocation power at the same time. This is important because the decision is highly dependent on power consumption. In the second algorithm, we introduce a decision criterion and therefore each user can decide on its own whether to offload or not. Therefore, 
The first algorithm is a centralized scheme, designed to be performed at the base station while the second one is a distributed scheme, which requires a partial information exchange, suitable to be performed at the user terminal. The complexity of these algorithms are also investigated.
\item[$\bullet$ ] Through simulations, we show that there exists an offloading region for each user where offloading can help  to save more power. By differentiating between cell edge users and normal users in the network, we show that the optimal region depends not only on a delay threshold and bit stream size of users but also on the position and channel conditions of the users. 
\end{itemize}
The rest of the paper is organized as follows. In Section~\ref{SectionII}, system model is presented. The problem formulation and the solution methodology are discussed in Section~\ref{SectionIV}. We propose our algorithms and corresponding complexity analysis in Section~\ref{SectionV} followed by the simulation results presented in Section~\ref{SectionVI}. Finally, we bring the concluding remarks in Section~\ref{SectionVII}.

\section{System Model and Problem Formulation} \label{SectionII}
\subsection{System Description}

According to \figurename \ref{SystemModel1}, we consider a cellular network with $ N_{c} $ base stations where mobile users (MUs) are uniformly distributed within a cell range. Each base station is equipped with a server which is responsible for the offloaded users' data processing and we assume there is a centralized unit which exchanges the required information between base stations using backhaul.
\begin{figure}
\centering
\includegraphics[scale = 0.40]{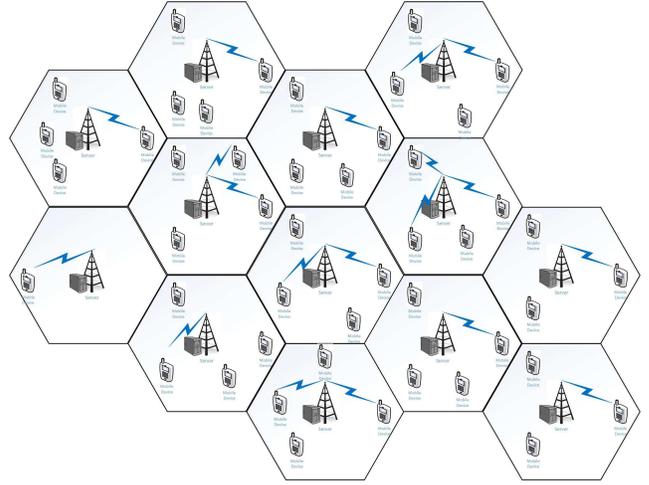}
\caption{System Model}
\label{SystemModel1}
\end{figure} 
Each cell can serve up to $ F_{i} $ active users. We assume that the available bandwidth $ B $ is divided to $ N $ sub-channels. The sub-channel model is adopted from \cite{access2010further} and 
is composed of large-scale fading, small-scale fading, and shadow fading. Also, we consider OFDMA as an access method; hence, users in the same cell cannot share the same sub-channel with each other. However, each user might experience an interference from neighboring cells. In this model, user $ j $ in cell $ i $ has a bit stream of size $ L_{i,j} $. %We have generated the users' bit stream size with uniform distribution between $ [\bar{\mathrm{L}}-\frac{1}{10} \bar{\mathrm{L}}, \bar{\mathrm{L}}+\frac{1}{10}\bar{\mathrm{L}}] $, where $ \bar{\mathrm{L}} $ is the average bit stream size. 
Users can process  data on their own or send it to the cloud. Users cannot use both schemes, e.g., sending a portion of  data to the cloud and processing the remaining data locally. Data corresponding to the user $ j $ in cell $ i $ should be processed within a time frame called maximum acceptable delay (delay threshold), $ T_{i,j} $.%, generated with with uniform distribution between $ [\bar{\mathrm{T}}-\frac{1}{10}\bar{\mathrm{T}}, \bar{\mathrm{T}}+\frac{1}{10}\bar{\mathrm{T}}] $, where $ \bar{\mathrm{T}} $ is the average acceptable delay. 
We assume that the processed data is short and the response time can be neglected \cite{zhang2013energy}. 

\textcolor{black}{Summary of notations can be found in Table \ref{tab:notations}
\begin{table}[ht]%\label{Tab:notations}
	\centering
	\bgroup
	\caption{Summary of Frequency Used Notations }
	\label{tab:notations}
	\def\arraystretch{1.3}
\begin{tabular}{ l l} \hline
	Notation & Description \\\hline
	$ MU $& Mobile users\\
	$ N_c $&  Number of cells \\ 
	$ F $ & Number of MUs per cell\\
	$ N $ & Number of subcarriers \\ 
	$ p_{Local}^{i,j} $& Power consumption of CPU of user $ j $ in cell $ i $\\ 
	$ L_{i,j} $& Data length of user $ j $ in cell $ i $\\ 
	$ T_{i,j} $& Delay requirement of user $ j $ in cell $ i $\\ 
	$ C_{i,j} $&Required Number of CPU cycles\\
	$ p_{tx}^{i,j} $& Transmission power of user $ j $ in cell $ i $\\ 
	$ P^{Total} $& Total power consumption\\ 
	$ R_{min} $&  Minimum required rate\\ 
	$s_{i,j}$& Offloading decision of user $ j $ in cell $ i $\\
\end{tabular}	
\egroup
\end{table}
}

\subsection{Power Model}

\subsubsection{ Local Processing Power Model}
\textcolor{black}{
Different mobile devices have different processing resources and also applications require different amount of processing. Therefore, each mobile device has a processing task that can be defined in terms of	
%The mobile devices may have different computation and energy resources, such as in the case of smart phones, tablets, and laptops.Each mobile user i has a computation task  [11x], [12x], which can be described in terms of
\begin{itemize}
\item the  length of input  file $ L_{i,j} $ (bit), to run a certain processing task related with a certain application, including system settings, program
codes, and input parameters;
\item the delay threshold $ T_{i,j} $, including processing delay;
\item the required number of  CPU cycles  to accomplish the
processing, $ C_{i,j} $;
\item the output (e.g., the computation result).
\end{itemize}
The values and statistics on $ L_{i,j} $ and $ C_{i,j} $ can be acquired by applying program profilers \cite{lyu2017multiuser} \footnote{The program profiler monitors the program parameters, such as, execution time, required memory,  thread CPU time, etc \cite{lyu2017multiuser}.}}. %othman2014survey,

When users are supposed to process their data locally, the CPU power consumption is dominant. It is composed of dynamic power, circuit power, and leakage power \cite{wang2016mobile}. The dynamic power as a dominating power in the CPU is a \textcolor{black}{superlinear function of execution frequency and required CPU cycles. The energy consumption of task completion at mobile device,  when operating at low voltage, is calculated as \cite{lyu2017multiuser},  
\begin{equation}
\label{eq:intitaloffpower}
E^{Local}_{i,j} = \zeta C_{i,j} {(f^{cpu}_{i,j})}^{m},
\end{equation}
where $ \zeta $ is effective switching capacitance depending on the chip architecture and  is in the order of  $ 10^{-11} $. $ C_{i,j} $ is the required CPU cycles for a given file size and is determined by the application \cite{zhang2013energy}.
For a given input size, $ C_{i,j} $ can be derived as, %\cite{miettinen2010energy}, 
\begin{eqnarray}
	C_{i,j} = L_{i,j} x
\end{eqnarray}
where $ x $ is a random variable with an
empirical distribution \cite{lorch2001improving}. The estimation of this distribution depends on the complexity of the application and since this estimation is beyond the scope of our study, the interested readers can study \cite{yuan2003energy,yang2012energy,yuan2006energy}.
Let the  CDF of $ x $ be $ F_x $ and the probability that CPU  completes the task within the delay threshold be $ \rho $. Then, given $ L_{i,j} $, we have 
\begin{eqnarray}
&F_{X} (C_{i,j}/L_{i,j}) &= \rho \\
&C_{i,j} &= L_{i,j} (F_{X}^{-1}(\rho)) \label{eq:freqcycle}
\end{eqnarray}
Note that the complexity of a task is reflected in the distribution function of  $ F_X $. Recalling from \eqref{eq:intitaloffpower}, the energy consumption of completing task becomes: 
\begin{eqnarray}
E^{Local}_{i,j} = \zeta \sum_{k =1}^{C_{i,j}}(1-F_{X}(k)) (f_{i,j}(k))^m \label{eq:energy_raw}
%L_{i,j} [F_{X}^{c^{-1}}(\rho)] f^m F_{X}^{c}(\rho) 
\end{eqnarray} 
where $ C_{i,j} $ is defined in \eqref{eq:freqcycle}. It is worth mentioning that the term within the summation in \eqref{eq:energy_raw} is the energy consumption of CPU cycle $ k $ with clock frequency $ f_{i,j} $, provided that the task has not been completed yet.
 The authors in \cite{zhang2013energy} derived the optimal value for CPU frequency by minimizing the energy consumption of device such that the  task is completed within the delay threshold $ T_{i,j} $.  Then, the minimum energy consumption of device with optimal CPU frequency is given by \cite{zhang2013energy},
 \begin{eqnarray}
 E^{*}_{i,j} = \frac{\zeta }{T^{m-1}_{i,j}}  \sum_{k =1}^{C_{i,j}}  \theta(k,\rho)^{m} \label{eq:opt_en}
 \end{eqnarray} 
 where $ \theta(k,\rho) $ is a function $ \rho $ and application complexity.  In \cite{zhang2013energy} (Appendix C), it is proven that the summation in \eqref{eq:opt_en} is proportional to the $ L^{m}_{i,j}  $. 
} Hence, the minimum power consumption of CPU is proportional to the \textcolor{black}{$ (L/T)^m $}, where $ T $ is the maximum acceptable delay and $ L $ is the users' bit stream size and $ m $ is the power of the scaling factor. Consequently, we use the following model for local processing power consumption: 

\begin{equation}
\label{Local_Power_Consumption_Model}
p^{Local}_{i,j} = M^{A}_{i,j} \frac{L^{m}_{i,j}}{T^{m}_{i,j}},
\end{equation}
where $ p^{Local}_{i,j} $ is the local processing power consumption of user $ j $ in cell $ i $ and $ M^{A}_{i,j} $ is a constant value for user $ j $, in cell $ i $, and application $ A $ that depends on the users' CPU and application parameters \textcolor{black}{, e.g., $ \rho $ and $ k  $, and is given in  \cite{zhang2013energy}. For the sake of simplicity we remove the superscript $ A $ in the rest of paper}. In our assumption, the users utilize the entire available time for the task completion.

\subsubsection{Offloading Power Model}
The transmission power for sending data to the cloud is: 
\begin{equation}
\label{User_Power_Consumption}
p^{tx}_{i,j} = \displaystyle \sum_{n=1}^{N} a_{i,j,n} p_{i,j,n},
\end{equation}
where $ p^{tx}_{i,j} $ denotes the transmission power consumption of user $ j $ in cell $ i $, \textcolor{black}{ $ p_{i,j,n} $ is transmission power of the same users on subchannel $ n $,} and $ a_{i,j,n} $ is a binary variable representing whether the corresponding sub-channel is assigned to the user or not. Therefore, the user's total transmission power is
\begin{equation}
\label{Offloading_Power_Model}
P^{Tx}_{i,j} = \frac{1}{\eta} p^{tx}_{i,j} + p_c,
\end{equation}
where $ \eta $ is power amplifier coefficient and $ p_c $ is a constant circuit power.

\subsubsection{ Aggregated Power Model}
Total power consumption of the active users in the network can be written as:
\begin{equation}
\label{total_power_consumption_pij}
{\bf{P}}^{Total}= \displaystyle \sum_{i=1}^{N_c} \displaystyle \sum_{j=1}^{F_i} p_{i,j}
\end{equation}
where
\begin{eqnarray}
\label{p_ij} \nonumber
p_{i,j} & = & s_{i,j} p^{Tx}_{i,j} + (1-s_{i,j}) p^{Local}_{i,j} \\ %\nonumber
%& = & s_{i,j} (\displaystyle \frac{1}{\eta}\sum_{n=1}^{N} a_{i,j,n} p_{i,j,n} + p_c) \\& + & (1-s_{i,j}) \frac{M_{i,j} L_{i,j}^m}{T^m_{i,j}}.
\end{eqnarray}
The integer variable $s_{i,j}$ takes the value of $0$ if user $j$ in cell $i$ uses its own processor and takes the value of $1$ if the user sends  data to the cloud. Therefore, the total power consumption can be written as: 
\begin{eqnarray}
\label{total_power_consumption} \nonumber
{\bf{P}}^{Total} & = & \displaystyle \sum_{i=1}^{N_c} \displaystyle \sum_{j=1}^{F_i} \displaystyle \sum_{n=1}^{N} s_{i,j} \frac{1}{\eta}(a_{i,j,n} p_{i,j,n} + p_c)\\ & + & \displaystyle \sum_{i=1}^{N_c} \displaystyle \sum_{j=1}^{F_i} (1-s_{i,j}) \frac{M_{i,j} L_{i,j}^m}{T^m_{i,j}}.
\end{eqnarray}
Moreover, the signal to noise plus interference ratio at the base station in cell $ i $ is given by:
\begin{equation}
\label{SINR}
\gamma_{i,j,n}=\frac{p^{tx}_{i,j,n} h_{i,j,n}}{\sigma^2 + I_{i,n}},
\end{equation}
where the channel gain from $ j $th MU of $ i $th cell is denoted by $ h_{i,j,n} $. The channel gain from user $ m $, in cell $ k $ to the cell $ i $ is denoted by $ h^{i}_{k,m,n} $. The first term in the denominator of (\ref{SINR}) is the noise power in bandwidth $ B $ and the second term is the interference from other cells on channel $ n $ in cell $ i $ which can be calculated as:
\begin{equation} \label{Interferenec}
I_{i,n}=\displaystyle \sum_{\substack{ k=1 \\ k \neq i}}^{N_c} \displaystyle \sum_{m=1}^{F_i} a_{k,m,n} s_{k,m} p^{tx}_{k,m,n} h_{k,m,n}^{i}.
\end{equation}
\textcolor{black}{
\subsection{Delay Model}
\subsubsection{Local Processing Delay Model}
When users are processing locally we assume that the processing task can be completed within the  delay threshold and hence $ T^{Local}\le T_{i,j} $. As the more stringent acceptable delay becomes, the more power is consumed at mobile device. We assume that the processing time at mobile device is equal to the users'  delay threshold.
\subsubsection{Offloading  Delay Model}
In case of offloading, the experienced delay of user $ j $ in cell $ i $ is composed of transmission delay in uplink, processing delay at the edge cloud, and  transmission delay in downlink. Therefore, total transmission delay is expressed as: 
\begin{eqnarray} \label{eq:Delay}
T^{Off}_{i,j} = T^{tx, up}_{i,j} + T^{Edge}_{i,j} +  T^{tx, dl}_{i,j}
\end{eqnarray} 
The first term on the right-hand side of (\ref{eq:Delay}) is the uplink transmission delay and can be calculated based on the uplink transmission rate. The second term is due to the processing delay at the edge cloud and can be calculated as \cite{lyu2017multiuser},
\begin{eqnarray}
T^{Edge}_{i,j} = \frac{L_{i,j}}{F^{Edge}},
\end{eqnarray}
where $ F^{Edge} $ is the computation processing allocated to the user.  The processing delay is contingent on the processing resources and the work load at the edge cloud. Reducing the processing delay is an active research area \cite{liu2018edge} however it is out of scope of this paper. If the edge cloud is equipped with powerful processing resources, $ F^{Edge} $ becomes very large and hence for small $ L_{i,j} $, $ T^{Edge}_{i,j} $ becomes negligible compared to the transmission time. For instance, assume that the required CPU cycles be in the range of $ [1-1000] M cycles $, and the CPU frequency  is $ 10 GHz $ \cite{lyu2017multiuser}, then the processing time at the edge cloud will be in $ [0.05-5] msec $. 
The last term in (\ref{eq:Delay}), is the return transmission time and is determined by the length of response message as well as the downlink transmission rate. As per \cite{lyu2017multiuser,you2017energy,li2019energy} this downloading time is negligible due to the much smaller size of the resulting message assuming some health care  and face detection applications.
}

\subsection{Problem Formulation} \label{SectionIII}
In this section, we develop the mathematical formulation for decision making and resource allocation problem. The base station determines the offloading users and allocates sub-channels to its users and specifies the suitable power level on each sub-channel. The objective of the resource allocation is to minimize the aggregated power consumption of all users by allocating resources to the offloading users in a way that their delay requirement is satisfied. In our model we assume that users utilize the maximum available time. The optimization problem can be formulated as follows:

\begin{align*}
\label{Original_problem}
\min_{\{{\bf{a}},{\bf{p}},{\bf{s}}\}} & \hspace{1.5cm} {\bf{P}}^{Total}\numberthis \\
& \hspace{-0.5cm}\text{subject to }\\
\text{C1: } &\begin{matrix}
0 \leq p_{i,j}^{\text{Tx}} \leq p_{max}, & \forall i,j, &
\end{matrix} \\
%\text{C2: } &\begin{matrix}
%\displaystyle \sum_{\substack{ k=1 \\ k \neq i}}^{N_c} \displaystyle \sum_{j=1}^{F_i} s_{k,j} a_{k,j,n} p_{k,j,n} %h_{k,j,n}^{i} \leq I_{th}^{(n)}, & \forall i,n, &
%\end{matrix} \\
\text{C2: } &\begin{matrix}
T^{tx}_{i,j} \leq T_{i,j}, & \forall i,j, &
\end{matrix} \\
\text{C3: } &\begin{matrix}
\displaystyle \sum_{j=1}^{F_i} \displaystyle \sum_{n=1}^{N} s_{i,j} a_{i,j,n} \log_{2} (1+\gamma_{i,j,n}) \leq R_{\textcolor{black}{i,} max}^{Proc}, & \forall i, &
\end{matrix} \\
\text{C4: } &\begin{matrix}
\displaystyle \sum_{j=1}^{F_i} a_{i,j,n} \leq 1, & \forall i,n, &
\end{matrix} \\
\text{C5: } &\begin{matrix}
a_{i,j,n} \in \{0,1\}, & \forall i,j,n, &
\end{matrix} \\
\text{C6: } &\begin{matrix}
s_{i,j} \in \{0,1\}, & \forall i,j. &
\end{matrix}
\end{align*}

In \eqref{Original_problem}, the objective is to minimize the total power consumption of all active MUs in the network. The constraint C1 indicates that the transmit power of each user is limited to $p_{max}$. %Interference among the MUs connected to different base stations play an vital role in the offloading decision \cite{mach2017mobile}. The constraint C2 states that for each base station $i$, the interference arising from other cells on each sub-channel is restricted to be within a threshold. 
In constraint C2, $ T^{tx}_{i,j} $ is the transmission time. This constraint guarantees \textcolor{black}{ that user $j$ in the  cell $i$ completes} its transmission before \textcolor{black}{ the delay threshold}, e.g., $ T_{i,j} $. If a user decides to process  data locally, then the CPU will be responsible for satisfying this constraint.  In our analysis we assume that CPU uses the entire available time. It is worth mentioning  that  $ T_{i,j} $ is given by the application and is an input to our problem. Later we will discuss how to relate this term to the rate.  The delays coming from processing \textcolor{black}{ at the edge} and \textcolor{black}{ downlink} transmission are not considered. The former is because of powerful processors and the latter is due to short response message size of the processed data \cite{zhang2013energy}. In our model, the processing capacity of the cloud server is limited with the total arrival rate of bit streams coming from different users at the same time. C3 captures this processing capacity. C3 guarantees that if a user decides to transmit, there will be enough processing capacity\textcolor{black}{, e.g., $R_{i, max}^{Proc}$} , for serving it. It does not have any effect on the transmission capacity but on users' decision. In other words, users offload their processing task only if there is enough capacity for the processing.  The constraint C4 guarantees the OFDMA assumption in each cell where each sub-channel is assigned to at most one user. The constraints C5 and C6 indicate that the sub-channel and data offloading indices are binary variables. It is worth mentioning that the constraint C2 can be written in an equivalent form. Using C2 we will have

\begin{equation}
\label{Delay_equivalent}
\frac{L_{i,j}}{T^{tx}_{i,j}} \geq \frac{L_{i,j}}{T_{i,j}}.
\end{equation} 

Defining $R_{min} \triangleq \frac{L_{i,j}}{T_{i,j} B}$ and noting that the left side of \eqref{Delay_equivalent} is the total normalized data rate of the $j$-th user in the $i$-th cell, we obtain the following equivalent constraint for C2:
\begin{equation}
\label{constraint_C3}
s_{i,j} \displaystyle \sum_{n=1}^{N} a_{i,j,n} \log_{2} (1+\gamma_{i,j,n}) \geq s_{i,j} R_{min}, \forall i,j.
\end{equation} 
In the rest of this paper, we consider the constraint C2 in the form presented in \eqref{constraint_C3}.

The optimization problem defined in \eqref{Original_problem} is a mixed integer nonlinear problem (MINLP)  and cannot be solved optimally in a polynomial time. The non-convexity is due to different reasons in the problem, \textcolor{black}{ for instance, the binary inherent of decision making variable (constraint C6) and the combinatorial nature of sub-channel allocation (constraint C5). Another reason of non-convexity is} due to  constraints C2 and C3 and existence of the power allocation variable in the denominator of SINR formula defined in \eqref{SINR}. 
\textcolor{black}{The offloading decisions are binary variables and they not only depend on the delay constraints, but also on the others’ decisions, channel quality, and limited radio resources. Therefore, the problem \eqref{Original_problem}  can be mapped to a special maximum cardinality bin packing problem which is NP-hard \cite{chen2015efficient}. Note that the optimal solution can be obtained by enumerating and comparing all possible computation offloading decisions (exhaustive search method) or  available tools that can solve MINLP problem \cite{belotti2013mixed}. However, due to the complexity of the problem, it is not possible to solve the problem in polynomial time.  Therefore, we focus on solving this complex problem in polynomial time.  } In the following section, we address how to deal with the aforementioned challenges and solve the problem by converting it into a convex form. 
%Note that the optimal solution can be obtained by enumerating and comparing all possible computation offloading decisions (exhaustive search method) or  available tools that can solve MINLP problem [48]. However, due to the complexity of the problem, it is not possible to solve the problem in polynomial time.  Therefore, we focus on solving this complex problem in polynomial time. 
\section{Solution Methodology}\label{SectionIV}
In this section, we aim to transform the primary problem defined in Section \ref{SectionIII} into a canonical convex form. In this regard, we classify the challenges into two categories, binary variables and non-convex functions. 

To resolve the challenges caused by the binary variables, one approach is to relax the troublesome constraints, sub-channel allocation for instance, to shape the problem into a convex form and then making a hard decision in the end as we did in \cite{masoudi2017energy}. An alternative approach is to add auxiliary constraints to enforce the solution to be in our desired form as we will describe later. Another approach is to break the problem into sub-problems so that one could successively first solve the problem for the  binary variable and consequently, given this variable, solve the rest of the problem.

To deal with the non-convex functions, we utilize a theory of optimization for a superclass of convex functions, called Difference of Convex (D.C.) functions \cite{kha2012fast}. Later we demonstrate that our problem can be written in form of D.C. functions. In the end, applying Taylor approximation enables us to solve the last stage of converting the primary problem defined in \eqref{Original_problem}, into a convex form. Having all these powerful approaches available, we tackle the problem, as follows.

In the first step, we break down the problem into two sub-problems and then solve them successively. The first sub-problem is to determine the channel assignment for each user in each cell. The second sub-problem is to find out the decision variable and power allocation. We use the solution of the first sub-problem as an input to the second sub-problem. Also, the results of \textcolor{black}{the} second sub-problem \textcolor{black}{are} used to update the solution for the first sub-problem and this process continues until the convergence. Furthermore, we apply two approaches to solve the second sub-problem. The overview of two utilized approaches to solve the problem can be seen in equations \eqref{method_1} and \eqref{method_2}.

In the first approach, after separating the sub-channel assignment, the problem can be solved jointly for other variables, e.g., power allocation and decision variable as follows: 
\begin{align*}
\label{method_1}
\overbrace{{\bf{a}}[0] \rightarrow ({\bf{p}}[0],{\bf{s}}[0])}^{\text{Initialization}} \rightarrow \ldots \rightarrow \overbrace{{\bf{a}}[t-1] \rightarrow ({\bf{p}}[t-1],{\bf{s}}[t-1])}^{\text{Iteration} \hspace{0.2cm} t-1} \\
\hspace{-1cm} \rightarrow \overbrace{{\bf{a}}[t] \rightarrow ({\bf{p}}[t],{\bf{s}}[t])}^{\text{Iteration} \hspace{0.2cm} t} \rightarrow \overbrace{ {\bf{a}}^{\star} \rightarrow ({\bf{p}}^{\star},{\bf{s}}^{\star})}^{\text{\textcolor{black}{Final} Solution}}. \numberthis
\end{align*} 
%The problem of power allocation and data offloading are solved jointly thus it has closer results to the optimal solution. Moreover, it is worth mentioning that this problem should be solved at the base stations that will be explained later.

In the second approach, we separate sub-channel assignment, power allocation, and decision variable from each other as follows:
\begin{align*}
\label{method_2}
\overbrace{{\bf{a}}[0] \rightarrow \tilde{{\bf{p}}}[0]\rightarrow{\bf{s}}[0]}^{\text{Initialization}} \rightarrow \ldots \rightarrow \overbrace{{\bf{a}}[t-1] \rightarrow \tilde{{\bf{p}}}[t-1]\rightarrow{\bf{s}}[t-1]}^{\text{Iteration} \hspace{0.2cm} t-1} \\
\hspace{-1cm} \rightarrow \overbrace{{\bf{a}}[t] \rightarrow \tilde{{\bf{p}}}[t]\rightarrow{\bf{s}}[t]}^{\text{Iteration} \hspace{0.2cm} t} \rightarrow \overbrace{ {\bf{a}}^{\star} \rightarrow \tilde{{\bf{p}}}^{\star}\rightarrow{\bf{s}}^{\star}}^{\text{\textcolor{black}{Final} Solution}}. \numberthis
\end{align*}

The main difference between these two approaches is that in the former, we jointly solve the problem of power allocation and decision making; However, in the latter, we divide the second sub-problem into two steps and solve each sub-problem individually. \textcolor{black}{In both algorithms, we  iterate until the convergence criterion is met, e.g., the difference value of the optimization parameter be less than $ \epsilon $, or reach to the maximum number of iterations.} In the following subsections,  we first deal with solving the first sub-problem followed by solving the second sub-problem by converting it into a convex from.

%\subsection{Sub-Problem One: Optimal Sub-channel Assignment for a Fixed Power Allocation and Decision Making}
\subsection{Sub-Problem One: Optimal Sub-channel Assignment}
Given the power allocation vector $ \tilde{{\bf{p}}}[t \mbox{--} 1] $, the optimal sub-channel assignment $ \bf{a} [t] $ for further power allocation and offloading in the next iteration $ t $ is as follows:
\begin{prop}\label{SubChannel}
Given the power vector, minimum power consumption is attained when each sub-channel in each cell is assigned to the MU with the highest effective interference on that sub-channel.
\end{prop}
\begin{proof}
Because the problem is power minimization and also minimum data rate requirement of users should be satisfied, the minimum power is consumed when the inequality of minimum required rate becomes the equality. Now let us assume that all users are given the best possible channel to reach their data rate with minimum power consumption. Also, let a user have a channel with effective interference value lower than a highest value and the user has data rate $ r_{min} $ on that channel. Thus, the consumed power on that channel \textcolor{black}{can be derived by solving the following equation for power:}
\begin{eqnarray} 
\log_{2} (1+\gamma_{i,j,n}) & = & r_{min} 
\end{eqnarray}
\textcolor{black}{where $ \gamma_{i,j,n} $ is defined in (\ref{SINR}) and solving for the transmit power yields to:}
\begin{eqnarray} 
\label{Effective_Interference}
p^{tx}_{i,j,n} & = & \frac{\textcolor{black}{2^{r_{min}}-1}}{\frac{h_{i,j,n}}{\sigma^2 + I_{n}}},
\end{eqnarray}
 Also from our assumption, we know that the effective interference in a denominator of (\ref{Effective_Interference}), e.g., $\frac{h_{i,j,n}}{\sigma^2 + I^{n}} $, is not the highest possible value. Hence, if we assign the highest effective interference value to this user, the total power consumption will be lower and this is in contrast with the assumption of minimum power consumption. Therefore, minimum power is consumed when maximum effective interference is the criterion for the channel allocation. In other words, with higher effective interference, less power is consumed to satisfy the minimum required rate. 
\end{proof}
Let $ EI_{i,j,n} $ be the SINR on the channel $ n $   \textcolor{black}{assuming unit transmission power}.  High $ EI_{i,j,n} $  on a channel means that the MU \textcolor{black}{can achieve high SINR with a good channel condition and low interference from other cells. Channel assignment is made based on this metric.} Therefore, the decision for channel allocation will be made based on the following criterion: 
\textcolor{black}{
\begin{eqnarray} \label{Channel_Allocation}
\tilde{a}_{i,j,\tilde{n}} = 1 \Big|_{\tilde{n} = argmax_{n} (EI_{i,j,n}) } \forall i,j.
\label{ChannelAssignmentCriteria}
\end{eqnarray}
}
Thus, a channel allocation matrix $ \boldsymbol{a}[t] $ at time $ t $, can be formed with the elements obtained from the equation (\ref{Channel_Allocation}). 

At this stage we have solved the first sub-problem and the results will be available for next steps. In the next two subsections, we solve the second sub-problem introduced in \eqref{method_1} and \eqref{method_2}.

\subsection{Sub-Problem Two: Power Allocation, and Decision Making}
In the previous subsection we have solved the problem of sub-channel assignment and therefore one of the challenges of the primary problem \eqref{Original_problem} is resolved. The results of previous subsection will be used in this section to solve the sub-problem of power allocation and decision making. As in \eqref{method_1} and \eqref{method_2}, two approaches are applied to tackle the challenges. These approaches are discussed in the following subsections.
\subsubsection{Joint Power Allocation and Decision Making (J-PAD)} \label{JPAD_Prob}
%\subsubsection{ Optimal Power Allocation and Data Offloading for a Fixed Sub-channel Assignment}
Given a sub-channel assignment, the problem of joint power allocation and data offloading can be rewritten as:

\begin{align*}
\label{main_problem_assuming_fixed_subchannel_assignment}
%\hspace{0.1cm} \displaystyle \sum_{i=1}^{N_c} \displaystyle \sum_{j=1}^{F_i} \displaystyle \sum_{n=1}^{N} s_{i,j} p_{i,j,n} + \displaystyle \sum_{i=1}^{N_c} \displaystyle \sum_{j=1}^{F_i} (1-s_{i,j}) \frac{M L_{i,j}^3}{T^{3}_{i,j}} 
\min_{\{{\bf{p}},{\bf{s}}\}} & P^{Total} \numberthis \\
& \hspace{-0.5cm}\text{subject to }\\
\text{C1: } & \begin{matrix}
0 \leq s_{i,j} p_{i,j}^{Tx} \leq p_{max},&&\forall i,j,
\end{matrix} \\
%\text{C2: } &\begin{matrix}
%\displaystyle \sum_{\substack{ k=1 \\ k \neq i}}^{N_c} \displaystyle \sum_{j=1}^{F_i} s_{k,j} p_{k,j,n} h_{k,j,n}^{i} \leq I_{th}^{(n)}, && \forall i,n,
%\end{matrix} \\
\text{C2: } & \begin{matrix}
s_{i,j} \displaystyle \sum_{n=1}^{N} \log_{2} (1+\gamma_{i,j,n}) \geq s_{i,j} R_{min}, && \forall i,j,
\end{matrix} \\
\text{C3: } &\begin{matrix}
\displaystyle \sum_{j=1}^{F_i} \displaystyle \sum_{n=1}^{N} s_{i,j} \log_{2} (1+\gamma_{i,j,n})  \leq R_{\textcolor{black}{i,} max}^{Proc}, && \forall i,
\end{matrix} \\
\text{C6: } &\begin{matrix}
s_{i,j} \in \{0,1\}, && \forall i,j.
\end{matrix}
\end{align*} 

To solve \eqref{main_problem_assuming_fixed_subchannel_assignment}, we first reformulate it to a more mathematically tractable form. Since $s_{i,j}$ is a binary variable, we can write $s_{i,j} \log_{2} (1+\gamma_{i,j,n}) = \log_{2} (1+s_{i,j} \gamma_{i,j,n})$. 
Moreover, the problem consists of the product terms of $s_{i,j} p_{i,j,n}$. We use the following change of variable
\begin{equation}
\label{change_of_variable_power}
\tilde{p}_{i,j,n}=s_{i,j}p_{i,j,n},
\end{equation}
to recast the optimization problem. Also, the optimization problem includes integer variable $s_{i,j}$. Hence to convert $s_{i,j}$ into continuous variables, we can express the constraint C6 as the intersection of the following regions:
\begin{align*}
\label{new_regions}
& \begin{matrix}
\mathcal{R}_1: 0 \leq s_{i,j} \leq 1, \forall j,i,
\end{matrix}\\
& \begin{matrix}
\mathcal{R}_2: \sum_{j} \sum_{i} \left(s_{i,j}- s_{i,j}^2\right) \leq 0. &\\
\end{matrix}\numberthis
\end{align*}
\textcolor{black}{With $ \mathcal{R}_1 $ we limit $ s_{i,j} $ to be in the interval of [0,1] and with $ \mathcal{R}_2 $ we enforce $ s_{i,j} $ to approach either 0 or 1, since $  \left(s_{i,j}- s_{i,j}^2\right) $  takes  always non-negative values but the constraint pushes $  \left(s_{i,j}- s_{i,j}^2\right) $  to be non-positive. }

Hence, we can write the optimization problem of \eqref{main_problem_assuming_fixed_subchannel_assignment} as follows
\begin{align*}
\label{main_format_2}
& \min_{\tilde{{\bf{p}}},{\bf{s}}} {\bf{P}}^{Total}\\
& \text{s.t. } \text{C1--\textcolor{black}{C3}}, \mathcal{R}_1,\mathcal{R}_2. \numberthis
\end{align*}
The problem  \eqref{main_format_2} is a continuous optimization problem with respect to all variables. However, we aim to find integer solutions for $s_{i,j}$'s. To attain this goal, we add a penalty function to the objective function of \eqref{main_format_2}
to penalize it if the values of $s_{i,j}$'s are not integer. Thus, the problem can be modified to 
\begin{align*}
\label{penalized}
& \min_{\tilde{{\bf{p}}},{\bf{s}}} \hspace{0.5cm} \mathcal{L}(\tilde{{\bf{p}}},{\bf{s}},\lambda)\\
& \text{s.t. } \text{C1--\textcolor{black}{C3}},\mathcal{R}_1. \numberthis
\end{align*}
In \eqref{penalized}, $\mathcal{L}(\tilde{{\bf{p}}},{\bf{s}},\lambda)$ is the \textcolor{black}{partial} Lagrangian of \eqref{main_format_2}, and is defined as
\begin{equation}
\mathcal{L}(\tilde{{\bf{p}}},{\bf{s}},\lambda) \triangleq {\bf{P}}^{Total}+\lambda \sum_{j} \sum_{i} \left(s_{i,j}- s_{i,j}^2\right),
\end{equation}
where $\lambda$ is the penalty factor which should be $\lambda \gg 1$. It can be shown that, for sufficiently large values of $\lambda$, the optimization problem of \eqref{penalized} is equivalent to \eqref{main_format_2} and attains the same optimal value. %\cite{che2014joint} 

\begin{prop}
For sufficiently large values of $\lambda$, the optimization problem of \eqref{penalized} is equivalent to \eqref{main_format_2}
\end{prop}
\begin{proof}
	Please see Appendix A.
\end{proof}

Now, \textcolor{black}{\eqref{main_format_2}} can be converted to the following problem 
\begin{align*}
\label{DC_joint_main_problem}
\min_{\{\tilde{{\bf{p}}},{\bf{s}}\}} & \hspace{0.1cm} \displaystyle \sum_{i=1}^{N_c} \displaystyle \sum_{j=1}^{F_i} \displaystyle \sum_{n=1}^{N} \tilde{p}_{i,j,n} + \displaystyle \sum_{i=1}^{N_c} \displaystyle \sum_{j=1}^{F_i} (1-s_{i,j}) \frac{M_{i,j} L_{i,j}^m}{T_{i,j}^m} \\
& +\lambda \big(\displaystyle \sum_{i}^{} \displaystyle \sum_{j}^{} (s_{i,j} - s_{i,j}^2)\big) \numberthis \\
& \hspace{-0.5cm}\text{subject to }\\
\text{C1: } & \begin{matrix}
0 \leq \displaystyle \sum_{n=1}^{N} \tilde{p}_{i,j,n} \leq s_{i,j} p_{max},&&\forall i,j,
\end{matrix} \\
%\text{C2: } &\begin{matrix}
%\displaystyle \sum_{\substack{ k=1 \\ k \neq i}}^{N_c} \displaystyle \sum_{j=1}^{F_i} \tilde{p}_{i,j,n} h_{k,j,n}^{i} \leq I_{th}^{(n)}, && \forall i,n,
%\end{matrix} \\
\text{C2: } & \begin{matrix}
\displaystyle \sum_{n=1}^{N} \log_{2} (1+\frac{\tilde{p}_{i,j,n} h_{i,j,n}}{\sigma^2 + \tilde{I}_{i,n}}) \geq s_{i,j} R_{min}, && \forall i,j,
\end{matrix} \\
\text{C3: } &\begin{matrix}
\displaystyle \sum_{j=1}^{F_i} \displaystyle \sum_{n=1}^{N} \log_{2} (1+\frac{\tilde{p}_{i,j,n} h_{i,j,n}}{\sigma^2 + \tilde{I}_{i,n}}) \leq R_{\textcolor{black}{i,} max}^{Proc}, && \forall i,
\end{matrix} \\
\text{C6: } &\begin{matrix}
s_{i,j} \in [0,1], && \forall i,j,
\end{matrix}
\end{align*} 
%where 
%\begin{equation}
% \label{modified_SINR}
% \tilde{\gamma}_{i,j,n}=\frac{\tilde{p}_{i,j,n} h_{i,j,n}}{\sigma^2 + \tilde{I}^{(n)}},
%\end{equation}
where $\tilde{I}_{i}^{(n)} \triangleq \displaystyle \sum_{\substack{ k=1 \\ k \neq i}}^{N_c} \displaystyle \sum_{m=1}^{F_i} \tilde{p}_{k,m,n} h_{k,m,n}^{i} $. We can write the objective function in \eqref{DC_joint_main_problem} as $f_1(\tilde{\bf{p}},{\bf{s}})-f_2(\tilde{\bf{p}},{\bf{s}})$, where $f_1(\tilde{\bf{p}},{\bf{s}}) \triangleq \displaystyle \sum_{i=1}^{N_c} \displaystyle \sum_{j=1}^{F_i} \displaystyle \sum_{n=1}^{N} \tilde{p}_{i,j,n} + \displaystyle \sum_{i=1}^{N_c} \displaystyle \sum_{j=1}^{F_i} ((1-s_{i,j}) \frac{M_{i,j} L_{i,j}^m}{T_{i,j}^m} +\lambda s_{i,j})$, and $f_2(\tilde{\bf{p}},{\bf{s}}) \triangleq \lambda \sum_{i=1}^{N_c} \displaystyle \sum_{j=1}^{F_i} s_{i,j}^2$ are two convex functions. In a similar way, for $\forall i,j$, we define $z_{i,j,n}(\tilde{{\bf{p}}})$ and $q_{i,j,n}(\tilde{{\bf{p}}})$ as 
\begin{equation}
\label{z_ijn}
z_{i,j,n}(\tilde{{\bf{p}}}) \hspace{-0.2em} \triangleq \hspace{-0.2em} \log_{2} \hspace{-0.2em} \bigg( \hspace{-0.3em}\tilde{p}_{i,j,n} h_{i,j,n} \hspace{-0.2em}+ \hspace{-0.2em} \displaystyle \sum_{\substack{ k=1 \\ k \neq i}}^{N_c} \hspace{-0.2em} \displaystyle \sum_{m=1}^{F_i} \hspace{-0.2em} \tilde{p}_{k,m,n} h_{k,m,n}^{i} \hspace{-0.12em}+\hspace{-0.1em} \sigma^2 \hspace{-0.1em} \bigg), 
\end{equation}

\begin{equation}
\label{q_ijn}
q_{i,j,n}(\tilde{{\bf{p}}}) \triangleq \log_{2} \bigg( \displaystyle \sum_{\substack{ k=1 \\ k \neq i}}^{N_c} \displaystyle \sum_{m=1}^{F_i} \tilde{p}_{k,m,n} h_{k,m,n}^{i} +\sigma^2 \bigg), 
\end{equation}
then, we can write constraints \textcolor{black}{C2} and \textcolor{black}{C3} as follows
\begin{align*}
&\begin{matrix}
\text{\textcolor{black}{C2}: } \hspace{0.3cm} Z_{i,j}(\tilde{{\bf{p}}}) - Q_{i,j}(\tilde{{\bf{p}}}) \geq s_{i,j} R_{min}, \forall i,j,
\end{matrix}\\
&\begin{matrix}
\text{\textcolor{black}{C3}: } \hspace{0.3cm} Q_{i}(\tilde{{\bf{p}}})-Z_{i}(\tilde{{\bf{p}}}) \geq - R_{\textcolor{black}{i,} max}^{Proc}, \forall i,
\end{matrix} \numberthis
\end{align*}
where $Z_{i,j}(\tilde{{\bf{p}}}) \triangleq \displaystyle \sum_{n=1}^{N} z_{i,j,n}(\tilde{{\bf{p}}})$, $Q_{i,j}(\tilde{{\bf{p}}}) \triangleq \displaystyle \sum_{n=1}^{N} q_{i,j,n}(\tilde{{\bf{p}}})$, $Z_{i}(\tilde{{\bf{p}}}) \triangleq \displaystyle \sum_{j=1}^{F_i} \displaystyle \sum_{n=1}^{N} z_{i,j,n}(\tilde{{\bf{p}}})$, and $Q_{i}(\tilde{{\bf{p}}}) \triangleq \displaystyle \sum_{j=1}^{F_i} \displaystyle \sum_{n=1}^{N} q_{i,j,n}(\tilde{{\bf{p}}})$ are concave functions.
%Functions $Z_{i,j}(\tilde{{\bf{p}}})$, $Q_{i,j}(\tilde{{\bf{p}}})$, $Z_{i}(\tilde{{\bf{p}}})$, and $Q_{i}(\tilde{{\bf{p}}})$ are concave. Hence, the constraints C3 and C4 are written as the difference of two concave functions. 
Therefore, the problem is in the form of the difference of two convex (concave) functions (D.C. programming) \cite{kha2012fast}. In D.C. programming, we start from a feasible initial point and iteratively solve the optimization problem. Let $k$ denote the iteration number. At the $k$-th iteration, to make the problem convex, using the first order Taylor approximation for $f_2(\tilde{\bf{p}},{\bf{s}})$, $Q_{i,j}(\tilde{{\bf{p}}})$ and $Z_{i}(\tilde{{\bf{p}}})$ as follows
\begin{align*}
&\begin{matrix}
\tilde{f}_2(\tilde{{\bf{p}}},{\bf{s}}) \hspace{-0.1em}\approxeq \hspace{-0.1em} f_2(\tilde{\bf{p}},{\bf{s}}^{k-1})\hspace{-0.1em} + \hspace{-0.1em} \nabla_{{\bf{s}}}f^T_2(\tilde{{\bf{p}}}, {\bf{s}}^{k-1}). ({\bf{s}}-{\bf{s}}^{k-1}),
\end{matrix}\\
&\begin{matrix}
\tilde{Q}_{i,j}(\tilde{{\bf{p}}}) \hspace{-0.1em} \approxeq \hspace{-0.1em} Q_{i,j}(\tilde{{\bf{p}}}^{k-1}) \hspace{-0.1em}+\hspace{-0.1em} \nabla_{\tilde{\bf{p}}}Q^T_{i,j}(\tilde{{\bf{p}}}^{k-1}). (\tilde{{\bf{p}}}-\tilde{{\bf{p}}}^{k-1}),
\end{matrix}\\
&\begin{matrix}
\tilde{Z}_{i}(\tilde{{\bf{p}}}) \approxeq Z_{i}(\tilde{{\bf{p}}}^{k-1})+ \nabla_{\tilde{\bf{p}}}Z^T_{i}(\tilde{{\bf{p}}}^{k-1}). (\tilde{{\bf{p}}}-\tilde{{\bf{p}}}^{k-1}), \end{matrix} \numberthis
\end{align*}
%\begin{equation}
% \tilde{f}_2(\tilde{{\bf{p}}},{\bf{s}}) \approxeq f_2(\tilde{\bf{p}},{\bf{s}}^{k-1})+ \nabla_{{\bf{s}}}f^T_2(\tilde{{\bf{p}}}, {\bf{s}}^{k-1}). %({\bf{s}}-{\bf{s}}^{k-1}),
%\end{equation}
%\begin{equation}
% \tilde{Q}_{i,j}(\tilde{{\bf{p}}}) \approxeq Q_{i,j}(\tilde{{\bf{p}}}^{k-1})+ \nabla_{\tilde{\bf{p}}}Q^T_{i,j}(\tilde{{\bf{p}}}^{k-1}). %(\tilde{{\bf{p}}}-\tilde{{\bf{p}}}^{k-1}),
%\end{equation}
%\begin{equation}
% \tilde{Z}_{i}(\tilde{{\bf{p}}}) \approxeq Z_{i}(\tilde{{\bf{p}}}^{k-1})+ \nabla_{\tilde{\bf{p}}}Z^T_{i}(\tilde{{\bf{p}}}^{k-1}). %(\tilde{{\bf{p}}}-\tilde{{\bf{p}}}^{k-1}).
%\end{equation}
where $\tilde{{\bf{p}}}^{k-1}$ and ${\bf{s}}^{k-1}$ are the solutions of the problem at $(k-1)$-th iteration and $\nabla_x$ denotes the gradient operation with respect to $x$. Thus, at the $k$-th iteration, instead of dealing with the problem of \eqref{main_problem_assuming_fixed_subchannel_assignment}, we solve the following convex problem
\begin{align*}
\label{DC_joint_final_problem}
\min_{\{\tilde{{\bf{p}}},{\bf{s}}\}} & \hspace{0.1cm} f_1(\tilde{{\bf{p}}},{\bf{s}}) - \tilde{f}_2(\tilde{{\bf{p}}},{\bf{s}}) \numberthis \\
& \hspace{-0.5cm}\text{subject to: } \hspace{0.25cm}\text{C1, C2, C7,}\\ 
\text{\textcolor{black}{C2}: } & \begin{matrix}
Z_{i,j}(\tilde{{\bf{p}}}) - \tilde{Q}_{i,j}(\tilde{{\bf{p}}}) \geq s_{i,j} R_{min}, && \forall i,j,
\end{matrix} \\
\text{\textcolor{black}{C3}: } &\begin{matrix}
Q_{i}(\tilde{{\bf{p}}})-\tilde{Z}_{i}(\tilde{{\bf{p}}}) \geq - R_{\textcolor{black}{i,} max}^{Proc}, && \forall i.
\end{matrix}
\end{align*}
It can be shown that the D.C. programming results in a sequence of feasible solutions that iteratively achieves better solutions than previous iteration until it converges. 

\begin{prop}
The D.C. programming results in a sequence of feasible solutions that iteratively decrease the total power consumption of the network. 
\end{prop}
\begin{proof}
Please see Appendix  B.
\end{proof}

\subsubsection{Channel assignment, Power Allocation, and Decision Making (C-PAD)}
Similar to subsection \ref{JPAD_Prob}, we assume that channel assignment vector is given based on proposition \ref{SubChannel}. Given sub-channel assignment, the optimization problem can be rewritten as:

\begin{align*}
%\hspace{0.1cm} \displaystyle \sum_{i=1}^{N_c} \displaystyle \sum_{j=1}^{F_i} \displaystyle \sum_{n=1}^{N} s_{i,j} p_{i,j,n} + \displaystyle \sum_{i=1}^{N_c} \displaystyle \sum_{j=1}^{F_i} (1-s_{i,j}) \frac{M L_{i,j}^3}{T^{3}_{i,j}} 
\min_{\{{\bf{p}}\}} & P^{Total} \numberthis \\
& \hspace{-0.5cm}\text{subject to }\\
\text{C1: } & \begin{matrix}
0 \leq s_{i,j} p_{i,j}^{Tx} \leq p_{max},&&\forall i,j,
\end{matrix} \\
%\text{C2: } &\begin{matrix}
%\displaystyle \sum_{\substack{ k=1 \\ k \neq i}}^{N_c} \displaystyle %\sum_{j=1}^{F_i} s_{k,j} p_{k,j,n} h_{k,j,n}^{i} \leq I_{th}^{(n)}, && \forall i,n,
%\end{matrix} \\
\text{C2: } & \begin{matrix}
s_{i,j} \displaystyle \sum_{n=1}^{N} \log_{2} (1+\gamma_{i,j,n}) \geq s_{i,j} R_{min}, && \forall i,j,
\end{matrix} \\
\text{C3: } &\begin{matrix}
\displaystyle \sum_{j=1}^{F_i} \displaystyle \sum_{n=1}^{N} s_{i,j} \log_{2} (1+\gamma_{i,j,n}) \leq R_{\textcolor{black}{i,} max}^{Proc}, && \forall i,
\end{matrix}
\label{main_cpad_problem}
\end{align*} 
By applying the method used in previous section we can formulate the problem as a D.C. programming optimization problem. In other words, similar to (\ref{ineq1}) and (\ref{ineq2}) we have: 
\begin{align*}
\text{C3: }& Z_{i,j}(\tilde{{\bf{p}}}^{t}) - Q_{i,j}(\tilde{{\bf{p}}}^{t}) \geq R_{min},&&\forall i,j&& \numberthis 
\end{align*}
\begin{align*}
\text{C4: }& Q_{i}(\tilde{{\bf{p}}}^{t}) - Z_{i}(\tilde{{\bf{p}}}^{t})\geq - R_{\textcolor{black}{i,} max}^{Proc},&&\forall i,&& \numberthis 
\end{align*}
Applying the first order Taylor approximation, the optimization problem can be written as 
\begin{align*}
%\hspace{0.1cm} \displaystyle \sum_{i=1}^{N_c} \displaystyle \sum_{j=1}^{F_i} \displaystyle \sum_{n=1}^{N} s_{i,j} p_{i,j,n} + \displaystyle \sum_{i=1}^{N_c} \displaystyle \sum_{j=1}^{F_i} (1-s_{i,j}) \frac{M L_{i,j}^3}{T^{3}_{i,j}} 
\min_{\{{\bf{p}}\}} & P^{Total} \numberthis \\
& \hspace{-0.5cm}\text{subject to }\\
\text{C1: } & \begin{matrix}
0 \leq s_{i,j} p_{i,j}^{Tx} \leq p_{max},&&\forall i,j,
\end{matrix} \\
%\text{C2: } &\begin{matrix}
%\displaystyle \sum_{\substack{ k=1 \\ k \neq i}}^{N_c} \displaystyle \sum_{j=1}^{F_i} s_{k,j} p_{k,j,n} h_{k,j,n}^{i} \leq I_{th}^{(n)}, && \forall i,n,
%\end{matrix} \\
\text{C2: } & \begin{matrix}\\
Z_{i,j}(\tilde{{\bf{p}}}^{t}) - \{Q_{i,j}(\tilde{{\bf{p}}}^{t-1}) + \nabla_{\tilde{\bf{p}}}Q^T_{i,j}(\tilde{{\bf{p}}}^{t}). (\tilde{{\bf{p}}}^{t}-\tilde{{\bf{p}}}^{t-1})\}&&\\ \\
\geq R_{min} & \hspace{-130pt} \forall i,j \end{matrix} \\
\text{C3: } & \begin{matrix}\\
Q_{i}(\tilde{{\bf{p}}}^{t}) - \{Z_{i}(\tilde{{\bf{p}}}^{t-1})+ \nabla_{\tilde{\bf{p}}}Z^T_{i}(\tilde{{\bf{p}}}^{t}). (\tilde{{\bf{p}}}^{t}-\tilde{{\bf{p}}}^{t-1})\}&&\\ \\
\geq - R_{\textcolor{black}{i,} max}^{Proc} && \hspace{-120pt} \forall i&& \end{matrix}
\label{cpad_problem}
\end{align*}

Given sub-channel assignment and power consumption vectors, offloading decisions can be made by users. 
Recall the power consumption of user $ j $ in cell $ i $ in (\ref{Local_Power_Consumption_Model}) and (\ref{Offloading_Power_Model}). Each user can compare offloading and local processing power consumption to make the decision $ s_{i,j} $ as follows:

%\[ f(n) =
%\begin{cases}
%n/2 & \quad \text{if } n \text{ is even}\\
%-(n+1)/2 & \quad \text{if } n \text{ is odd}\\
%\end{cases}
%\]
\[ s_{i,j} =
\begin{cases} \numberthis\label{OffloadingDM}
1 & \quad P^{Local}_{i,j} > P^{Tx}_{i,j} \\
0 & \quad P^{Local}_{i,j} \leq P^{Tx}_{i,j} \\
\end{cases}
\]

%\begin{equation}
%\label{Offloading_Power_Model2}
%P^{Tx}_{i,j} = \frac{1}{\eta} p^{tx}_{i,j} + p_c,
%\end{equation}

\section{Algorithm Design}\label{SectionV}
In this section, based on our solutions, we propose two tractable algorithms to solve the optimization problem in  polynomial time. 
 The first algorithm fits well to a situation where information of all cells are available at the centralized unit and base stations are in charge of performing the offloading algorithms. The second algorithm suits well when offloading algorithm is performed at MUs' sides and only partial information exchange is required between base stations.

\subsection{J-PAD Algorithm } \label{J-PAD_Algorithm}
Algorithm \ref{Algorithm1} performs \textbf{J}oint \textbf{P}ower \textbf{A}llocation and \textbf{D}ecision making and is called \textbf{J-PAD}. J-PAD is designed to solve the convex optimization problem presented in (\ref{DC_joint_final_problem}). Here, the key idea is to make decision and allocate power simultaneously, while channels are assigned beforehand.   Algorithm \ref{Algorithm1} represents the procedure of solving the optimization problem using J-PAD algorithm \textcolor{black}{which starts  from random initial points for power allocation, channel assignment, and decision variable vectors}. 

\begin{algorithm}
\caption{ Joint Power Allocation and Decision Making (\textbf{J-PAD}) Algorithm}
\label{Algorithm1}
\begin{algorithmic}[1] 
\item Initialize power, $ \bf{a} $, $ \bf{s} $, $ I_{max}, \lambda=1, $ and $ Counter = 0 $
\While{ $ Counter \leq I_{max} $} {}
\\ \textbf{Channel Allocation}
\item Calculate $EI_{i,j,n} $ based on (\ref{Channel_Allocation}) $\forall i,j,n$
\item Form $\boldsymbol{a}[t] $ based on $EI_{i,j,n} $
\\ \textbf{Power Allocation and Offloading Decision}
\For{i=1 to $ N_c $}{}
\begin{enumerate}[label=\itshape\alph*\upshape)] 
\item \hspace{11pt} Solve the problem (\ref{DC_joint_final_problem}) using interior point method \cite{boyd2004convex}
\item \hspace{11pt} Update Power Vector based on the solution of (\ref{DC_joint_final_problem})
\item \hspace{11pt} Update $ s_{k,u,n} $ according to the solution of (\ref{DC_joint_final_problem})
\end{enumerate} 
\EndFor 
\item Update $ \lambda $, $ Counter = Counter + 1 $
\item Centralized unit updates the $ I $ based on (\ref{Interferenec}) and sends this value back to the base stations.
\EndWhile 
\end{algorithmic}
\end{algorithm}

J-PAD algorithm is composed of two main sections, channel assignment, based on the equation (\ref{Channel_Allocation}), and power allocation and offloading decision. After performing the second part, the power vectors and offloading decisions are updated at each base station and will be sent to the centralized unit. Then the centralized unit  updates the interference value on each channel and sends them back to each base station for next iteration. The problem is solved at the base station where the offloading algorithm is performed.  

Besides, $ \lambda $ plays an important role in the performance of J-PAD algorithm. It is a penalty factor to punish the objective function for any value of offloading decision variable $ s $ that is not binary. Therefore $ \lambda $ should be large enough, e.g., $ 10^5, $ ($ \lambda \gg 1 $) \cite{che2014joint}, to penalize the objective. One can fix this value to a predetermined high value but here we first set the $ \lambda $ to a relatively low value ( $ \lambda > 1 $ ). In this case, we let the value of $ s $ be a real value in [0,1]. By applying this approach, in the first steps, we do not penalize the objective too much, so that the algorithm has more freedom at choosing $ s $ and power. Then in next iterations we tighten the condition on $ s $ by choosing larger $ \lambda $ to make sure that the final value for $ s $ is binary. 
The benefit for using such a method is to deal with the trade-off between convergence and achieving optimal value. With a small value for $ \lambda $, \textcolor{black}{J-PAD} algorithm has more freedom to choose a non-binary value for offloading decision and therefore can find the value for power easier. With larger value for $ \lambda $, the algorithm focuses on penalizing the objective function due to non-binary offloading decision value. Because we have already found out the power value (in previous iterations with smaller $ \lambda $), the offloading decision variable converges fast to the binary value \cite{che2014joint}. 

\subsection{C-PAD Algorithm}

In this section, we propose an alternative algorithm to J-PAD which has less complexity and  the decision making process can be moved to the MUs’ side instead of the BS. In this situation, MUs only need partial information from other cells.

To avoid the integer inherent of the problem, we assign channels and make offloading decision iteratively before allocating the power. Hence, we divide the algorithm into three main parts: $ 1) $ Channel allocation which is done based on the Eq. (\ref{Channel_Allocation}), $ 2) $ Offloading decision which is performed by comparing the alternative solutions power consumption according to the Eq. \eqref{OffloadingDM}, and $ 3) $ Power allocation. In the latter part, channel allocation and offloading decisions are not optimization variables anymore because they are known for each user beforehand. Therefore, this algorithm performs \textbf{C}hannel allocation, \textbf{P}ower \textbf{A}llocation and \textbf{D}ecision making iteratively and is called \textbf{C-PAD} algorithm. The procedure of finding the solution with C-PAD algorithm is presented in Algorithm \ref{Algorithm2}.

\begin{algorithm}
\caption{ Channel allocation, Power Allocation and Decision Making (\textbf{C-PAD}) Algorithm }
\label{Algorithm2}
\begin{algorithmic}[1] 
\item Initialize initial points, $ I_{max}, \lambda, $ and $ Counter = 0 $
\While{ $ Counter \leq I_{max} $} {}
\\ \textbf{Channel Allocation}
\item Calculate $EI_{i,j,n} $ based on (\ref{Channel_Allocation}) $\forall i,j,n$
\item Form $\boldsymbol{a}[t] $ based on $EI_{i,j,n} $
\\ \textbf{Offloading Decision}
\item Determine the offloading decision based on (\ref{OffloadingDM}) for each user.
\item Update the channel allocation and offloading decision vector.
\\ \textbf{Power Allocation}
\For{i=1 to $ N_c $}{}
\begin{enumerate}[label=\itshape\alph*\upshape)] 
\item \hspace{15pt} Solve the problem (\ref{DC_joint_main_problem}) with a given channel allocation and offloading decision vector using interior point method \cite{boyd2004convex}
\item \hspace{15pt} Update Power Vector based on the solution found from (\ref{DC_joint_main_problem})
\end{enumerate} 
\EndFor 
\item $ Counter = Counter + 1 $
\item Centralized unit updates the parameter $ I $ based on (\ref{Interferenec}) and sends this value back to the base stations and base stations distribute it to the users.
\EndWhile 
\end{algorithmic}
\end{algorithm}

In  Algorithm \ref{Algorithm2} the channel allocation scheme is the same as Algorithm \ref{Algorithm1}. For offloading section, each user compares its power consumption for two possible cases, e.g., local processing or offloading and makes decision accordingly. Given \textcolor{black}{the power consumption values}, the problem of power minimization can be solved.
% By expelling the offloading decision from power allocation we are further away from optimal point. However, 
This segmentation enables us to perform the algorithm at users' side. In other words, the second algorithm is a distributed scheme with very low data exchange requirements at the expense of losing optimality. Centralized unit sends information about interference to each base station and the base stations relay this information to the users. Afterwards, users can use their local information and make their decisions. The procedure will continue until the convergence criteria is met. The computational complexity of the proposed algorithms will be discussed and compared in the next section.

\subsection{Complexity Analysis}
In this section, we investigate the computational complexity of our proposed algorithms. In both J-PAD and C-PAD, to assign sub-channels to the users, we have to find the user with highest effective interference. Let $F$ denote the maximum number of users existing in a cell, i.e., $F= \max\limits_{i=1,\ldots, N_c} F_{i}$. Since finding the maximum of a set with $K$ elements requires $O(K)$ operations, the sub-channel assignment phase has the complexity order of $O(NFN_{c})$. For  data offloading and power allocation in J-PAD algorithm, we have totally $N_c F (N+1)$ decision variables and $N_c (3F + N +1)$ convex and linear constraints. %\cite{che2014joint}.
 Therefore, the computational complexity of solving a joint data offloading and power allocation problem is given by \[O((N_c F (N+1))^3 (N_c (3F + N +1))) \approx O(N_c^4 F^3 N^3 (3F + N )) \]
%If $I_{T1}$ and $I_{T2}$ denote the maximum number of iterations required for the convergence of approach 1 and 2, respectively, the computational complexity of our proposed approaches are as follows.
%\[ O(I_{T1} (N_c^4 F^3 N^3 (3F + N )+N F N_c ) ) \], 
%\[ O(I_{T2} (N_c^4 F^3 N^3 (2F + N )+N F N_c + N_c F ) ) \].

In C-PAD algorithm,  data offloading and power allocation are separated. To find a data offloading strategy, it is sufficient to compare the power consumption in cases that each user uses its processor or sends its data to the cloud and select the one with lowest power consumption. Since, we have to carry this out for all users in all cells, we need $O(N_c F)$ operations. For the power allocation, we have totally $N_c F N$ variables and $N_c (2F+N+1)$ linear and convex constraints. Similar to what has been presented for the first approach, the power allocation computational complexity has the order of $O(N_c^4 F^3 N^3 (2F+N))$. 
%If $I_{T1}$ and $I_{T2}$ denote the maximum number of iterations required for the convergence of approach 1 and 2, respectively, the computational complexity of our proposed approaches are as follows.
%\[ O(I_{T1} (N_c^4 F^3 N^3 (3F + N )+N F N_c ) ) \], 
%\[ O(I_{T2} (N_c^4 F^3 N^3 (2F + N )+N F N_c + N_c F ) ) \].
The computational complexity of proposed methods is summarized in table \ref{table_1}. 

\begin{table}[ht]
\caption{Computational Complexity of proposed approaches.}
\label{table_1}
\centering
\scriptsize
\addtolength{\tabcolsep}{-3pt}
\begin{tabular}{| c | p{1.8cm} | l | l |}
\hline
% & \multicolumn{4}{c||}{\bf{Sub-channel values }} & \multicolumn{4}{c|}{\bf{Power values}} \\
& {\bf{Sub-channel Assignment}} & {\bf{Data Offloading}} & {\bf{Power Allocation}} \\
\hline
{\bf{J-PAD}} & $O(NFN_{c})$ & \multicolumn{2}{c|}{$O \bigg(N_c^4 F^3 N^3 (3F + N ) \bigg)$} \\ 
\hline
\bf{C-PAD} & $O(NFN_{c})$ & $O(N_c F)$ & $O \bigg(N_c^4 F^3 N^3 (2F+N)\bigg)$ \\
\hline
\end{tabular}
\addtolength{\tabcolsep}{3pt}
\end{table}

\textcolor{black}{The computational complexity reported  in Table \ref{table_1} is the worst case time complexity. %In fact, this computational complexity is for the entire algorithm, however, parts of the algorithm can be solved in parallel. 
Running such algorithms requires both time and energy which are functions of the number of required operations and processing speed of devices. However, the required time 
	\footnote{\textcolor{black}{For a user running C-PAD locally, we need to operate roughly $ (N_c F N)^3 (2F+N)$ operations. Using given values in Table \ref{table_2}, e.g., $N= 25 $ $N_c= 7 $  $F=5 $, roughly we need about $ 23.4e9 $ operations. Based on the \cite{FLOPs}, today's mobile phone CPUs can support about 2 Tera operation per second which suggest that  the delay due to running the algorithm is  within 10 msec. Note that this delay may limit the generality and operability of the algorithm for very stringent delay thresholds.  In our study, we assume that the delay threshold of each application is counted after running the algorithms. This is an acceptable assumption, since the range of considered delay is between $100 ms$ and $500 ms$ according to the scenarios we target.}} 
and 
energy \footnote{ \textcolor{black}{According to \cite{carroll2010analysis}, different usage scenarios consume energy between $ (.3 - 1)\tau$  $ Joule$, where $ \tau $ is the running time. Based on our calculation  the algorithm is running within about $ \tau =10$ msec  which results in $ (3 - 10)$   $m Joule$ energy consumption.}}
 to run our proposed algorithms are negligible compared to the delay thresholds and total power consumption of devices.}

\section{Simulation Results} \label{SectionVI} 
\subsection{System Parameters}
In this section we evaluate the performance of the proposed algorithms using numerical studies after defining the system parameters and baseline cases.
The scenario as depicted in \figurename \ref{SystemModel1} is a multi-cell mobile network where each base station is equipped with a computing server. The simulation parameters and their corresponding values are summarized in Table \ref{table_2}. We assume that each cell can serve up to $ F_i $ users and their QoS is defined as a maximum acceptable delay. The carrier frequency is set to $ 2 GHz $ and thermal noise is considered as a zero mean Gaussian random variable with a variance of $ \sigma^2 $ and a power spectral density of $ N_{0} =$~$-174 dbm/Hz $, so $ \sigma^{2} = (W/N)N_0 $. The pathloss model is adopted from \cite{access2010further}, shadow fading is modeled as zero mean log normal distribution with variance of $ 8 db$, and Rayleigh fading is modeled as a unit-mean exponential distribution. Each cell has a coverage radius of $ 500m $ and users are distributed uniformly within a cell coverage. 

\begin{table}[ht]
\centering
% \scriptsize
% \addtolength{\tabcolsep}{-3pt}
\bgroup
\caption{Simulation parameter values} \label{table_2}
\def\arraystretch{1.2}% 1 is the default, change whatever you need
\begin{tabular}{ m{4.5cm} l l }
\hline
\textbf{Definition} & \textbf{Notation} & \textbf{Value}\\ \hline
Sub-carrier bandwidth & $ B $& $ 200 $ $ KHz $ \\ 
Number of sub-carrier & $ N $ & $ 25 $ \\ 
Number of cells & $ N_c $ & $ 7 $ \\ 
Number of active MUs & $ F_i $ & $ 5 $ \\ 
Circuit power & $ P_c $ & $ 100$ $mW $ \\ 
Power amplifier efficiency & $ \eta $ & $ 0.4 $ \\ 
Scaling factor power & $ m $ & $ 3 $ \\ 
%Maximum allowable interference & $ I_{th} $& $ -101 $ $ dbm $ \\ 
Noise power spectral density & $ N_{0} $& $ -174 $ $ dbm/Hz $ \\
Maximum transmit power of users & $ P_{max} $& $ 23 $ $ dbm $ \\ 
Maximum delay of user j in cell i & $ T_{i,j} $ & $ 100 $ $ ms $ \\ 
Average bit stream size of user j in cell i & $ L_{i,j} $ & $ 2000 $ $ bits $ \\ 
% Macrocell radius &$ R_{M} $& $ 1000 $ $ m $ \\ 
% Femtocell radius &$ R_{F} $& $ 20 $ $ m $ \\ 
% & \multicolumn{4}{c||}{\bf{Sub-channel values }} & \multicolumn{4}{c|}{\bf{Power values}} \\
% & {\bf{Sub-channel Assignment}} & {\bf{Data Offloading}} & {\bf{Power Allocation}} \\
% \hline
% {\bf{Approach I}} & $O(NFN_{c})$ & \multicolumn{2}{c|}{$O \bigg(N_c^4 F^3 N^3 (3F + N ) \bigg)$} \\ 
% \hline
% \bf{Approach II} & $O(NFN_{c})$ & $O(N_c F)$ & $O \bigg(N_c^4 F^3 N^3 (2F+N)\bigg)$ \\
% \hline
\end{tabular}
% \addtolength{\tabcolsep}{3pt}
\egroup
\end{table}

%\subsection{Baselines}
We have compared our results with two baseline cases to understand the main reasons behind the power savings: whether the saving is  dominated more by the offloading decisions or it stems from power control on each channel. In the first one, all MUs use local processing and nobody offloads  data to the cloud. Comparing with this scheme, we can observe how much power saving can be obtained by utilizing the proposed algorithms.  The
second baseline is equal power allocation. In this scheme, power is equally allocated on user's assigned channels such that the required QoS is satisfied. For equal power allocation, starting from zero, we increase the power level and check the constraints until all of them are satisfied. According to the given power allocation, channel assignment is performed based on \eqref{Channel_Allocation}. Through the comparison of these schemes, we can find out the amount of power saving related to the power adjustment on each channel. 
\subsection{Simulation Results}
\begin{figure}[t!]
\centering
\includegraphics[scale = .45]{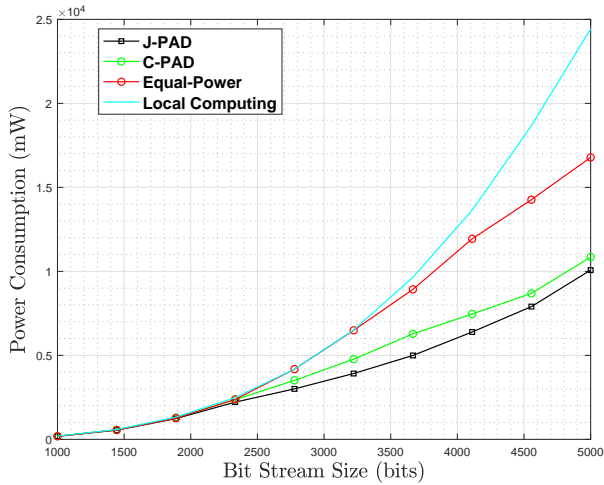}%  FinalFig1
\caption{ Aggregate power consumption for different bit stream sizes }
\label{streamAggregatePower}
\end{figure}
The  power consumption of all users over different bit stream sizes is depicted in \figurename \ref{streamAggregatePower}. The larger \textcolor{black}{the} bit stream size, the more power is consumed meanwhile the gap between  power consumption of local computing and proposed algorithms increases.

\begin{figure}[t!]
\centering
\includegraphics[scale = .6]{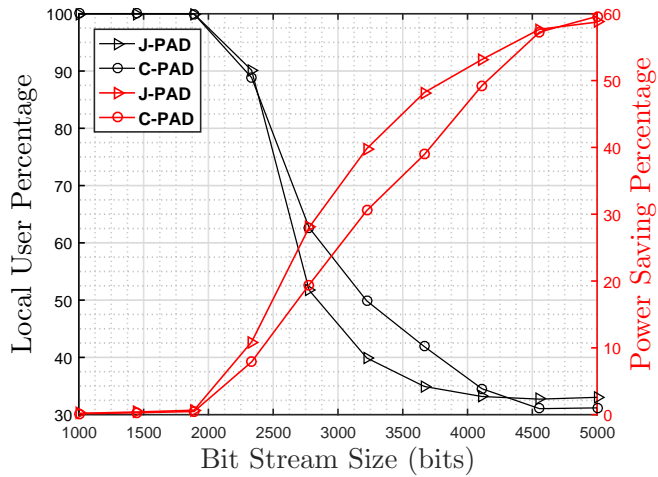}% fig3
\caption{Power saving  against Local Computing and  percentage of local computing users }
\label{streamPSaving}
\end{figure}

\figurename \ref{streamPSaving} illustrates how J-PAD and C-PAD could help users to offload and how much power is saved. As can be seen from the figure, by increasing the bit stream size, the percentage of local computing users decreases. The reason is that local processing of the large  bit stream size results in higher power consumption in comparison with sending  data to the cloud. Therefore, confirmed by simulations depicted in \figurename \ref{streamPSaving}, users tend to use the alternative option, e.g., offloading, to save power. For large bit stream sizes, using J-PAD and C-PAD, about $ 30\% $ of users decide on local processing and $ 60\% $ power saving is attained in comparison with local computing base line. Comparing J-PAD and C-PAD, J-PAD slightly outperforms C-PAD in terms of power saving at most $ 20\% $, while it has more complexity. 

\begin{figure}[t!]
\centering
\includegraphics[scale=.6]{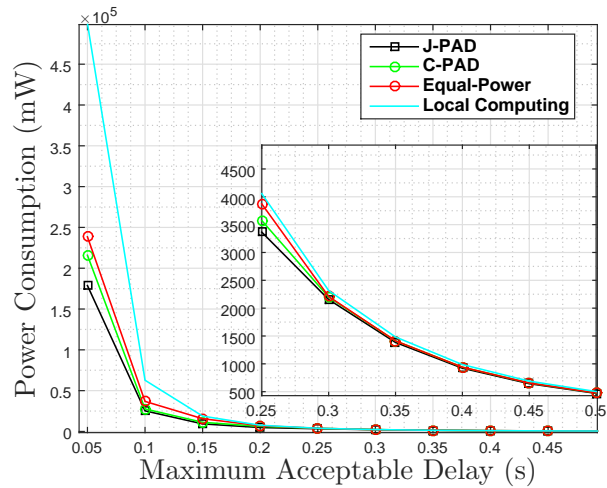}%FinalFig3
\caption{Aggregate power consumption for different acceptable delays}
\label{delayAggregatePower}
\end{figure}

The maximum acceptable delay as a QoS requirement is another parameter that affects the power consumption and offloading decisions. In \figurename \ref{delayAggregatePower}, we have investigated the impact of delay  on our algorithms. Longer acceptable delay for offloading users means lower data rate requirement and consequently lower power consumption for  data transmission. Also for local computing users, this delay results in a lower power consumption confirmed by the model. The gap between the proposed algorithms and the benchmark is wide at the beginning and becomes tighter as maximum acceptable delay gets longer. To discover why, we have illustrated the percentage of the local computing users and the corresponding power savings in \figurename \ref{delayPSaving}. For short delays, the power consumption is relatively high which, decays as the acceptable delay becomes longer. It can be seen that users, based on the mobile devices power model, prefer local processing when they can tolerate the considerable delay. Using J-PAD or C-PAD for short delays, about $ 65 \% $ power saving is obtained while for long delays, i.e., when the processing delay requirement can be relaxed, power savings from offloading is diminished because local processing power drops down exponentially with the processing delay due to the power model in Eq. \eqref{Local_Power_Consumption_Model}, resulting in almost all users processing locally. 

%It is worth noting that the mobile device power model plays an important role in this offloading decision. As can be seen in \figurename \ref{powerModelDiff}, from a certain point in the X axis (delay), some models consume more power while they have consumed less power before that. Therefore, users' decision highly depends on the adopted power model.

\begin{figure}[t!]
\centering
\includegraphics[scale = .6]{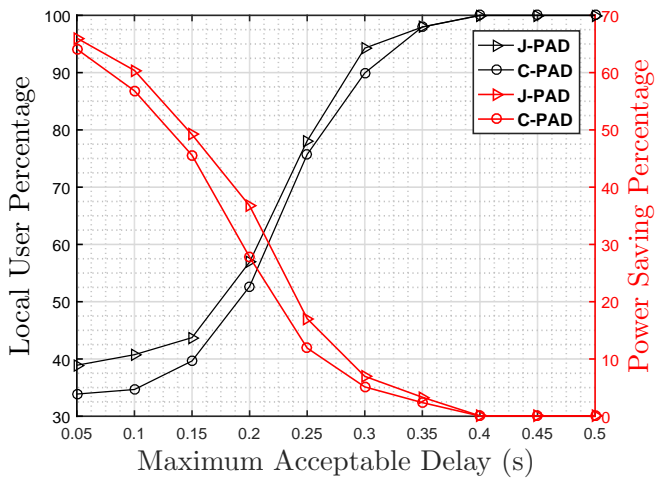}% fig5
\caption{Power saving against Local Computing and percentage of local computing users for different acceptable delays}
\label{delayPSaving}
\end{figure}
%Percentage savings against Local Computing
%\begin{figure}[t!]
%\centering
% \includegraphics[width=3.5in]{SystemModel}
%\includegraphics[scale = .6]{DelayEffect}
%\caption{ Aggregate power consumption of different local computing power models }
%\label{powerModelDiff}
%\end{figure}

Users' offloading decision not only relies on the power model but also depends on the other users' decisions due to the interference coming from neighboring cells. Consequently, the number of active users in the network is also crucial. \figurename \ref{userAggregatePower} and \figurename \ref{userPSaving} address this issue. The more users  that exist in the network, the higher interference  created, which means lower SINR, lower data rate and consequently more experienced delay for the users. As a result, the percentage of local users (not necessarily the absolute number) increases with the increasing number of users.

In order to investigate the impact of interference, we relaxed the interference constraint and solved the problem optimally in \figurename \ref{userAggregatePower}. \textcolor{black}{In case of small number of users, due to lower interference imposed by others, the interference level is negligible and hence  it has no impact on the solution. However, as the number of users increases,  interference becomes critical and consequently considering the interference-free scenario results in wrong decision.  Implementation of the interference-free scenario, not only reveals the impact of interference but also provides us a lower bound on the optimal solution of the problem. This solution is lower bound since the interference is not considered and hence with lower power consumption, the rate constraints are satisfied.  For small number of users, this lower bound coincides with our solution. For large number of users, however, this solution is not feasible, it can provide us a lower bound on the optimal solution. Please note that the 'Local Computing' simulation is an upper bound for the problem since all users can process their data locally in the worst case.  For instance for high load, e.g., $ 10 $ users in  a cell, proposed algorithms could still achieve about $ 30\% $ power saving in comparison with local computing base line.}

%In order to investigate the impact of interference, we relaxed the interference constraint and solved the problem optimally in \figurename \ref{userAggregatePower}. \textcolor{black}{In case of small number of users, due to small imposed interference from others, the interference level is negligible and hence  it has no impact the solution. However, as the number of users increases,  interference becomes critical and consequently ignoring the interference results in wrong decision.  The result of "No Interference" simulation, not only reveals the impact of interference but also gave us a lower bound on the optimal solution of the problem. This solution is lower bound since the interference is not considered and hence with lower power consumption, the rate constraints are satisfied.  For small number of users, this lower bound coincides with our solution. For large number of users, however, this solution is not feasible, it is still a lower bound on the optimal solution. Please note that the "All Local" simulation is an upper bound for the problem since all users can process their data locally in the worst case. } For $ 10 $ users in cell, proposed algorithms could still achieve about $ 30\% $ power saving in comparison with local computing base line.

\begin{figure}[t!]
\centering
\includegraphics[scale = .46]{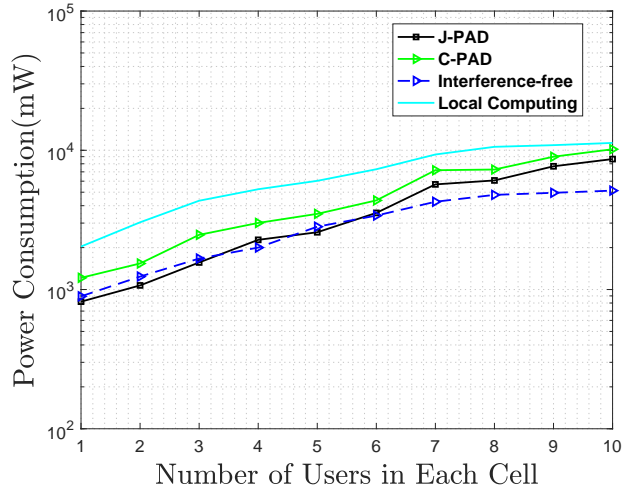}% fig6, UserNumber1
\caption{Aggregate power consumption over different number of users }
\label{userAggregatePower}
\end{figure}
\begin{figure}[t!]
\centering
\includegraphics[scale = .6]{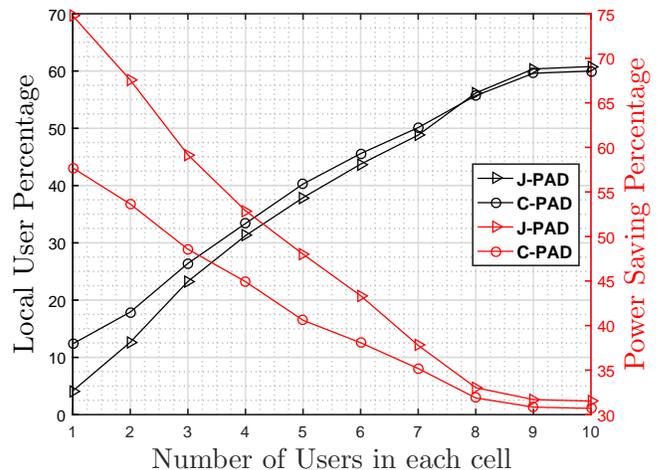}% fig7
\caption{ Power saving against Local Computing and percentage of local computing users for different number of users}
\label{userPSaving}
\end{figure}

\begin{figure*}[t!]
\begin{subfigure}[]{0.5\textwidth}
\centering
\includegraphics[scale = .35]{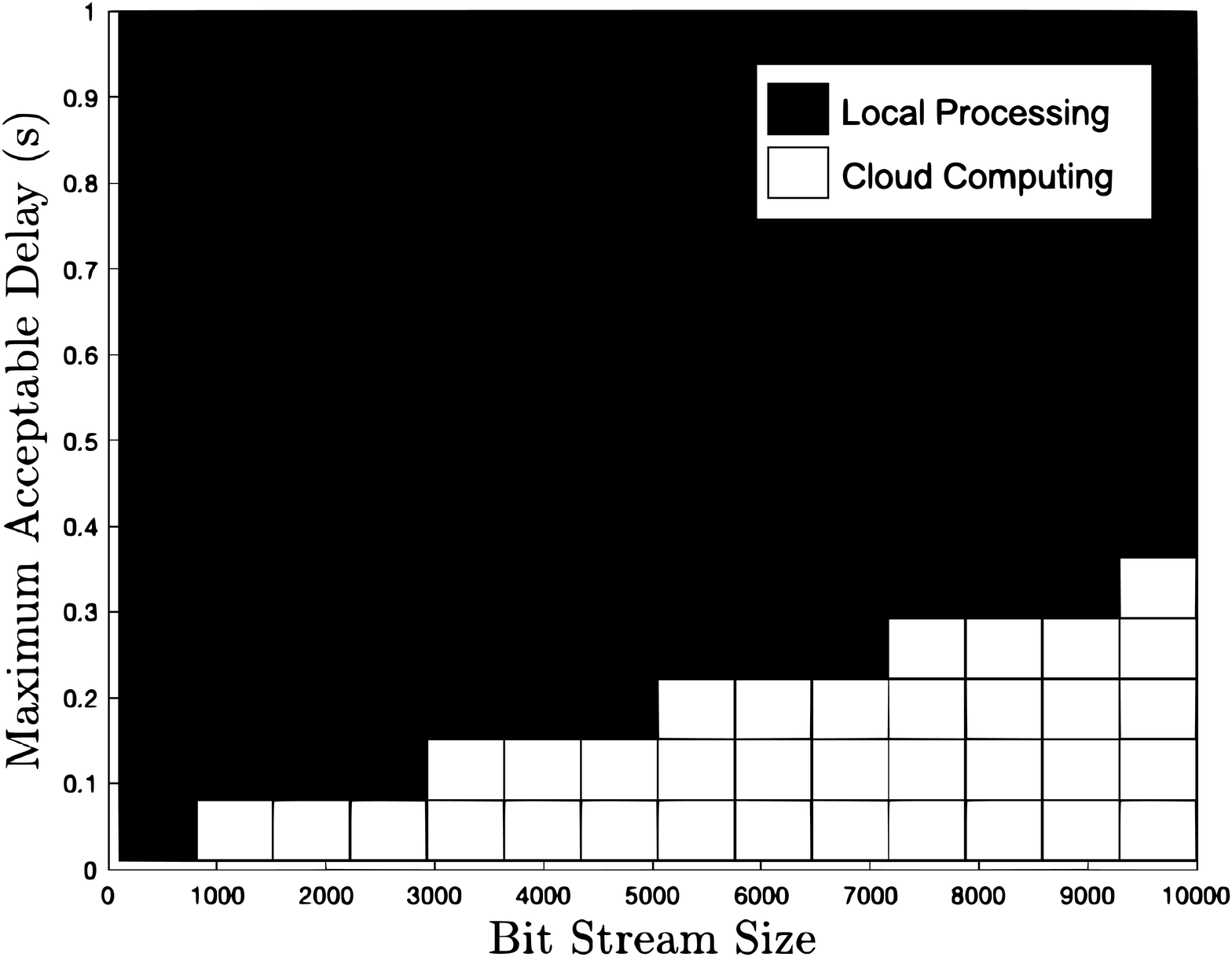}
\caption{J-PAD offloading regions for normal user }
\label{regionAlg1Normal}
\end{subfigure}
\begin{subfigure}[]{0.5\textwidth}
\centering
\includegraphics[scale = .35]{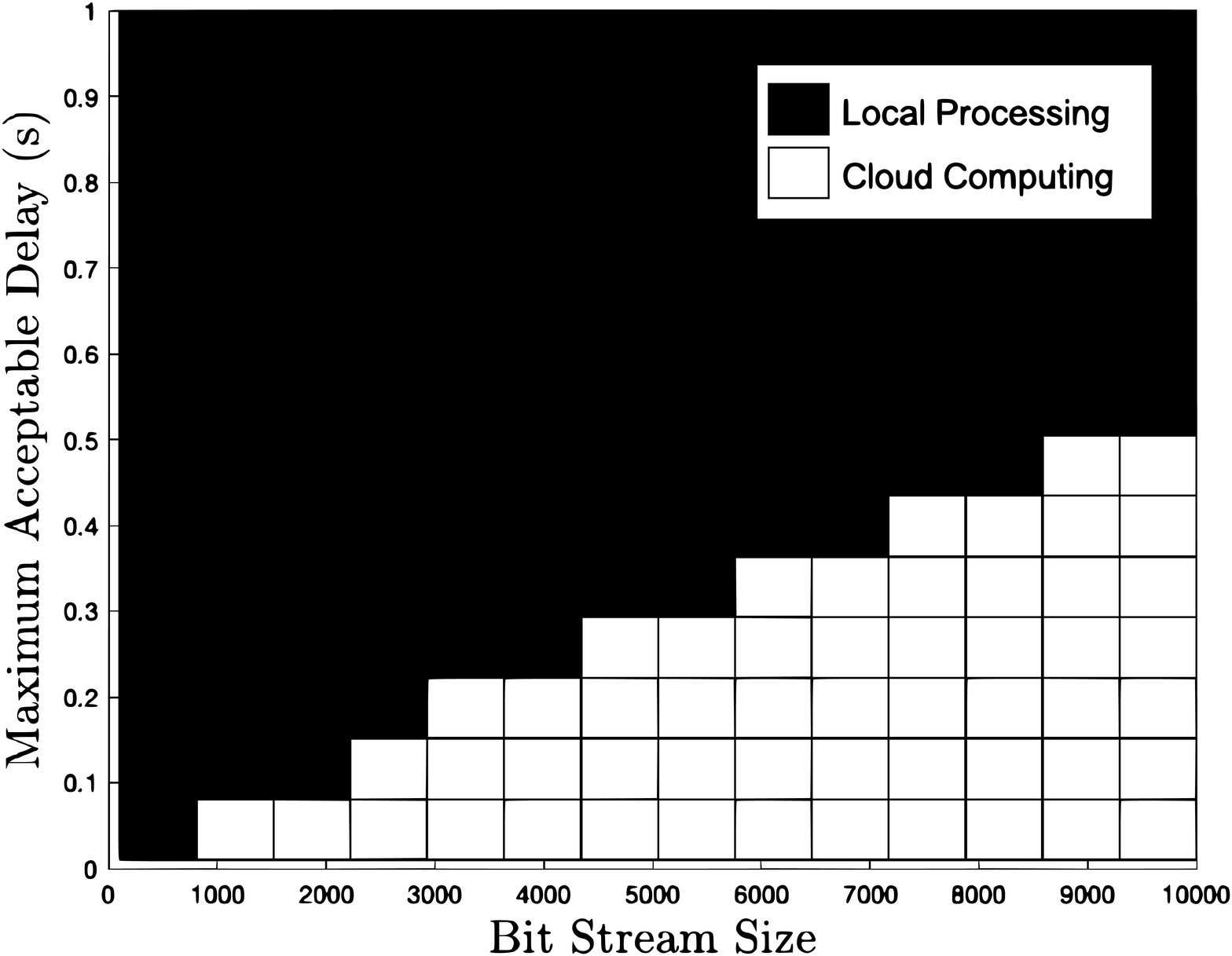}
\caption{C-PAD offloading regions for normal user }
\label{regionAlg2Normal}
\end{subfigure}
% \begin{subfigure}[]{0.5\textwidth}
% \centering
% % \includegraphics[width=3.5in]{SystemModel}
% \includegraphics[scale = .4]{CloudRegion3}
% \caption{}
% \label{regionAlg1edge}
% \end{subfigure}
% \begin{subfigure}[]{0.5\textwidth}
% \centering
% % \includegraphics[width=3.5in]{SystemModel}
% \includegraphics[scale = .4]{CloudRegion4}
% \caption{}
% \label{regionAlg2edge}
% \end{subfigure}
\caption{Offloading regions for J-PAD and C-PAD}
\label{offloadRegion}
\end{figure*}
In \figurename \ref{offloadRegion}, we investigate the offloading region for normal and cell-edge users to find out when offloading can save power. 

For normal users, non-cell edge users, one can see that for large bit stream size and low acceptable delay, e.g., yellow region in the figures, J-PAD and C-PAD can help mobile devices to save power. For fixed delay, by enlarging the bit stream size we end in  the offloading region. Moreover, C-PAD has a wider region than J-PAD because the offloading decision is made before solving the optimization problem. Our simulation results also reveal that cell edge users with poor channel gain and SINR cannot  benefit from offloading to the cloud. Because with bad channel condition, users need more power than local processing to send data to the cloud at an acceptable rate to meet the delay requirements. Here, providing users with better SINR, e.g., using joint transmission, might be helpful.
%
%\begin{figure*}[!htp]
% \begin{minipage}{.25\textwidth}
% \centering
% \includegraphics[scale = .32 ]{CloudRegion1}
% \caption{Maximum allowable delay effect on power consumption}
% \label{Fig_DelayPower}
% \vspace{-0.45cm}
% \end{minipage}
% \begin{minipage}{.25\textwidth}
% \centering 
% \includegraphics[scale = .32]{CloudRegion2}
% \caption{Maximum allowable delay effect on the locally executing users}
% \label{Fig_DelayUserPerc}
% \vspace{-0.45cm}
% \end{minipage}
% \begin{minipage}{.25\textwidth}
% \centering
% \includegraphics[scale = .3]{CloudRegion3}
% \caption{The Optimal region of typical user for offloading}
% \label{Fig_OptRegion}
% \vspace{-0.45cm}
% \end{minipage}
%\begin{minipage}{.25\textwidth}
% \centering
% \includegraphics[scale = .3]{CloudRegion4}
% \caption{The Optimal region of typical user for offloading}
% \label{Fig_OptRegio}
% \vspace{-0.45cm}
%\end{minipage}
%\end{figure*}
%

\ifCLASSOPTIONcaptionsoff
\newpage
\fi

\section{Conclusion} \label{SectionVII}
In this paper, the power minimization of mobile devices as a crucial aspect of mobile edge computing networks is considered. Accordingly, an optimization problem aimed at minimizing the power consumption of all users is formulated. To take into account the users' quality of service, maximum tolerable delay of users are considered. Knowing the inherent non-convexity of
our primary problem, we applied the D.C. approximation to
transform the non-convex problem to a convex one. We proposed two algorithms, called J-PAD and C-PAD to solve the problem in polynomial time. The J-PAD algorithm is better than C-PAD in terms of power saving but at the cost of complexity; therefore, it is not suitable to be used in the mobile terminal but in the BSs with high processing resources. C-PAD has the advantage of running at the users' side at the cost of losing the optimality. 
Our simulations demonstrated that there exists an offloading region for non-cell edge users where they can benefit from offloading data to the cloud.
Finally, confirmed by our results, significant enhancement in terms of power consumption of mobile devices could be achieved using the proposed algorithms.

% if have a single appendix:
%\appendix[Proof of the Zonklar Equations]
% or
%\appendix % for no appendix heading
% do not use \section anymore after \appendix, only \section*
% is possibly needed

% use appendices with more than one appendix
% then use \section to start each appendix
% you must declare a \section before using any
% \subsection or using \label (\appendices by itself
% starts a section numbered zero.)
%

\appendices
\section{Proof of Proposition 2}

	We start with this point that the optimization problem of \eqref{penalized} can be expressed as $\min\limits_{\tilde{{\bf{p}}},{\bf{s}}} \max \limits_{\lambda} \mathcal{L}(\tilde{{\bf{p}}},{\bf{s}},\lambda)$ and its dual problem can be written as $\max \limits_{\lambda} \min\limits_{\tilde{{\bf{p}}},{\bf{s}}} \mathcal{L}(\tilde{{\bf{p}}},{\bf{s}},\lambda)$. Suppose that $\tilde{{\bf{p}}}_{\lambda}$, ${\bf{s}}_{\lambda}$, and $\varphi (\lambda)$ denote the optimal solution and the optimal value of  the optimization problem of \eqref{penalized}, respectively, i.e.
	\begin{equation}
	\label{varphi_definition}
	\varphi (\lambda) = \mathcal{L}(\tilde{{\bf{p}}}_{\lambda},{\bf{s}}_{\lambda},\lambda)=\min\limits_{\tilde{{\bf{p}}},{\bf{s}}} \mathcal{L}(\tilde{{\bf{p}}},{\bf{s}},\lambda)
	\end{equation}
	Then, we will have
	\begin{align*}
	\label{max_varphi_definition}
	\max\limits_{\lambda} \varphi (\lambda) =\max\limits_{\lambda} \mathcal{L}(\tilde{{\bf{p}}}_{\lambda},{\bf{s}}_{\lambda},\lambda)= \max\limits_{\lambda} \min\limits_{\tilde{{\bf{p}}},{\bf{s}}} \mathcal{L}(\tilde{{\bf{p}}},{\bf{s}},\lambda) \\
	\leq \min\limits_{\tilde{{\bf{p}}},{\bf{s}}} \max\limits_{\lambda} \mathcal{L}(\tilde{{\bf{p}}},{\bf{s}},\lambda) \numberthis
	\end{align*}
	\textcolor{black}{and hence  problem \eqref{max_varphi_definition} becomes equivalent to \eqref{main_format_2}.}
	Recall that for $\forall ~{\bf{s}} \in \mathcal{D},\mathcal{R}_1$, we have
	\begin{equation*}
	\sum_{i} \sum_{j} \left( s_{i,j}- s_{i,j}^2\right) \geq 0.
	\end{equation*}
	In other words, $\varphi(\lambda)$ is an increasing function in $\lambda$ and according to \eqref{max_varphi_definition}, is bounded by the optimal value of problem \eqref{main_format_2}. If for some $0 \leq \lambda^* < \infty$, $\sum_{i} \sum_{j} \left( s_{i,j}- s_{i,j}^2\right) = 0$, then $(\tilde{{\bf{p}}}_{\lambda^*},{\bf{s}}_{\lambda^*})$ is feasible for the main problem, too. As a result, we will have
	\begin{align*}
	\label{aaa}
	\varphi(\lambda^*)=\mathcal{L}(\tilde{{\bf{p}}}_{\lambda^*},{\bf{s}}_{\lambda^*},\lambda^*) = \max\limits_{\lambda} \mathcal{L}(\tilde{{\bf{p}}}_{\lambda^*},{\bf{s}}_{\lambda^*},\lambda) \\
	\geq \min\limits_{\tilde{{\bf{p}}},{\bf{s}}} \max\limits_{\lambda} \mathcal{L}(\tilde{{\bf{p}}}_{\lambda},{\bf{s}}_{\lambda},\lambda) \numberthis
	\end{align*}
	comparing \eqref{aaa} and \eqref{max_varphi_definition}, we conclude that the strong duality holds and we have
	\begin{equation}
	\varphi(\lambda^*)=\max\limits_{\lambda} \varphi (\lambda),
	\end{equation}
	since $\varphi(\lambda)$ is a monotonically increasing function with respect to $\lambda$, for $ \forall \lambda \geq \lambda^*$ we have 
	\begin{equation}\label{eq:phi}
	\varphi(\lambda)={\bf{p}}^{Total} (\tilde{{\bf{p}}}^{\lambda^*}, {\bf{s}}^{\lambda^*})=\min\limits_{\tilde{{\bf{p}}},{\bf{s}}} \max\limits_{\lambda} \mathcal{L}(\tilde{{\bf{p}}},{\bf{s}},\lambda)
	\end{equation}
	\textcolor{black}{and hence  problem \eqref{eq:phi} becomes equivalent to \eqref{main_format_2}.}
	At the optimal point and for the second case where we have $\sum_{i} \sum_{j} \left(s_{i,j}- s_{i,j}^2\right) >0$, $\varphi(\lambda^{*})$ goes to $\infty$ because of the monotonicity of the $\varphi(\lambda)$ with respect to the $\lambda$. This contradicts the max-min inequality which states that $\varphi(\lambda^{*})$ is bounded from above. Thus, the term $\sum_{i} \sum_{j} \left(s_{i,j}- s_{i,j}^2\right)$ should be zero, and the results of the first case hold.

\section{Proof of Proposition 3 }
To show that our solutions are feasible for the original problem, first, we notice that the solution of the approximated problem in the $i$-th iteration must satisfy the constraint C2 and C3, i.e.,
\begin{align*}\label{ineq1}
& Z_{i,j}(\tilde{{\bf{p}}}^{t}) - \tilde{Q}_{i,j}(\tilde{{\bf{p}}}^{t})=&&&&\\
& Z_{i,j}(\tilde{{\bf{p}}}^{t}) - \{Q_{i,j}(\tilde{{\bf{p}}}^{t-1})+ \nabla_{\tilde{\bf{p}}}Q^T_{i,j}(\tilde{{\bf{p}}}^{t}). (\tilde{{\bf{p}}}^{t}-\tilde{{\bf{p}}}^{t-1})\}&&\\
& \geq R_{min}, \numberthis 
\end{align*}
\begin{align*}\label{ineq2}
& Q_{i}(\tilde{{\bf{p}}}^{t}) - \tilde{Z}_{i}(\tilde{{\bf{p}}}^{t})=&&&&\\
& Q_{i}(\tilde{{\bf{p}}}^{t}) - \{Z_{i}(\tilde{{\bf{p}}}^{t-1})+ \nabla_{\tilde{\bf{p}}}Z^T_{i}(\tilde{{\bf{p}}}^{t}). (\tilde{{\bf{p}}}^{t}-\tilde{{\bf{p}}}^{t-1})\}&&\\
& \geq - R_{max}^{Proc}, \numberthis 
\end{align*}
On the other hand, since $Z_{i,j}$ and $Q_{i}$ are two concave functions with respect to $\tilde{{\bf{p}}}$, due to the first order condition for the concave functions \cite{boyd2004convex}, we have
\begin{equation}
\label{first_order_condition_convex_Z_ij}
Z_{i,j}(\tilde{{\bf{p}}}) \leq Z_{i,j}(\tilde{{\bf{p}}}^{t-1}) + \nabla_{\tilde{{\bf{p}}}} Z_{i,j}^T (\tilde{{\bf{p}}}^{t-1}).(\tilde{{\bf{p}}}-\tilde{{\bf{p}}}^{t-1}).
\end{equation}
\begin{equation}
\label{first_order_condition_convex_Q_i}
Q_{i}(\tilde{{\bf{p}}}) \leq Q_{i}(\tilde{{\bf{p}}}^{t-1}) + \nabla_{\tilde{{\bf{p}}}} Q_{i}^T (\tilde{{\bf{p}}}^{t-1}).(\tilde{{\bf{p}}}-\tilde{{\bf{p}}}^{t-1}).
\end{equation}
Substituting $\tilde{{\bf{p}}}=\tilde{{\bf{p}}}^{t}$ into \eqref{first_order_condition_convex_Z_ij} and \eqref{first_order_condition_convex_Q_i} results in
\begin{equation}
Z_{i,j}(\tilde{{\bf{p}}}^t) \leq Z_{i,j}(\tilde{{\bf{p}}}^{t-1}) + \nabla_{\tilde{{\bf{p}}}} Z_{i,j}^T (\tilde{{\bf{p}}}^{t-1}).(\tilde{{\bf{p}}}^t-\tilde{{\bf{p}}}^{t-1}). 
\end{equation}
\begin{equation}
Q_{i}(\tilde{{\bf{p}}}^t) \leq Q_{i}(\tilde{{\bf{p}}}^{t-1}) + \nabla_{\tilde{{\bf{p}}}} Q_{i}^T (\tilde{{\bf{p}}}^{t-1}).(\tilde{{\bf{p}}}^t-\tilde{{\bf{p}}}^{t-1}).
\end{equation}
From \eqref{ineq1} and \eqref{ineq2}, we conclude that
\begin{align*}
& Z_{i,j}(\tilde{{\bf{p}}}^{t}) - Q_{i,j}(\tilde{{\bf{p}}}^{t}) \geq &&&&\\
& Z_{i,j}(\tilde{{\bf{p}}}^{t}) - \{Q_{i,j}(\tilde{{\bf{p}}}^{t-1})+ \nabla_{\tilde{\bf{p}}}Q^T_{i,j}(\tilde{{\bf{p}}}^{t}). (\tilde{{\bf{p}}}^{t}-\tilde{{\bf{p}}}^{t-1})\}&&\\
& \geq R_{min}, \numberthis 
\end{align*}
\begin{align*}
& Q_{i}(\tilde{{\bf{p}}}^{t}) - Z_{i}(\tilde{{\bf{p}}}^{t})\geq &&&&\\
& Q_{i}(\tilde{{\bf{p}}}^{t}) - \{Z_{i}(\tilde{{\bf{p}}}^{t-1})+ \nabla_{\tilde{\bf{p}}}Z^T_{i}(\tilde{{\bf{p}}}^{t}). (\tilde{{\bf{p}}}^{t}-\tilde{{\bf{p}}}^{t-1})\}&&\\
& \geq - R_{max}^{Proc}. \numberthis 
\end{align*}
Thus, the solution for the approximated problem is feasible for the original problem too. Now, we show that the total power consumption of the network will decrease iteratively. Since $g({\bf{s}})$ is a convex function, due to the first order condition for the convex functions \cite{boyd2004convex}, we have
\begin{equation}
\label{first_order_condition_convex_g}
g({\bf{s}}) \geq g({\bf{s}}^0) + \nabla_{{\bf{s}}} g^T ({\bf{s}}^{t-1}).({\bf{s}}-{\bf{s}}^{0}).
\end{equation}
Using \eqref{first_order_condition_convex_g} and considering the fact that the objective function of \eqref{DC_joint_main_problem} can be written as $f({\bf{s}})-g({\bf{s}})$, at the $(t+1)$-th iteration we have
\begin{align*}
\label{new}
& \begin{matrix}
f({\bf{s}}^{t+1})-g({\bf{s}}^{t+1}) \leq f({\bf{s}}^{t+1}) - g({\bf{s}}^{t}) -\nabla_{\hspace{-0.6mm}{\bf{s}}} g^T \hspace{-0.3mm} ({\bf{s}}^{t}\hspace{-0.5mm}).({\bf{s}}^{t+1}\hspace{-0.5mm}-\hspace{-0.5mm}{\bf{s}}^{t})\hspace{-0.5mm}& \\
\end{matrix} \\
& \begin{matrix}
= \min\limits_{{\bf{s}}} f({\bf{s}}) - g({\bf{s}}^{t}) -\nabla_{{\bf{s}}} g^T ({\bf{s}}^{t}).({\bf{s}}-{\bf{s}}^{t})
\end{matrix} \\
& \begin{matrix}
\leq f({\bf{s}}^{t}) - g({\bf{s}}^{t}) -\nabla_{{\bf{s}}} g^T ({\bf{s}}^{t}).({\bf{s}^{t}}-{\bf{s}}^{t})=f({\bf{s}}^{t}) - g({\bf{s}}^{t})
\end{matrix}. \\
\end{align*}
Thus, the total power consumption of the network decreases as iterations continue. 

%Appendix one text goes here.
%
%% you can choose not to have a title for an appendix
%% if you want by leaving the argument blank
%\section{}
%Appendix two text goes here.
%
%
%% use section* for acknowledgment
%\section*{Acknowledgment}
%
%
%The authors would like to thank...
%

% Can use something like this to put references on a page
% by themselves when using endfloat and the captionsoff option.
\ifCLASSOPTIONcaptionsoff
\newpage
\fi

% trigger a \newpage just before the given reference
% number - used to balance the columns on the last page
% adjust value as needed - may need to be readjusted if
% the document is modified later
%\IEEEtriggeratref{8}
% The "triggered" command can be changed if desired:
%\IEEEtriggercmd{\enlargethispage{-5in}}

% references section

% can use a bibliography generated by BibTeX as a .bbl file
% BibTeX documentation can be easily obtained at:
% http://www.ctan.org/tex-archive/biblio/bibtex/contrib/doc/
% The IEEEtran BibTeX style support page is at:
% http://www.michaelshell.org/tex/ieeetran/bibtex/
\bibliographystyle{IEEEtran}
% argument is your BibTeX string definitions and bibliography database(s)
\bibliography{IEEEabrv,MCCFinalTMC_WithoutColor}
\end{document}